\documentclass[review, 3p]{elsarticle}
\makeatletter
\def\ps@pprintTitle{%
 \let\@oddhead\@empty
 \let\@evenhead\@empty
 \def\@oddfoot{}%
 \let\@evenfoot\@oddfoot}
\makeatother

\newcommand\blue[1]{{\color{black}{#1}}}

\usepackage{url}
\usepackage{epsfig,epsf}
\usepackage{graphicx,graphics}
\usepackage[labelformat=simple]{subcaption}

\usepackage{amssymb,amsmath}
\usepackage{verbatim,enumerate}
\usepackage{verbatim,multicol,color}
\usepackage{multirow}
\usepackage{booktabs}
\usepackage{multicol}
\usepackage[english]{babel}
\usepackage{blindtext}
\usepackage{float}
\usepackage{bm}
\usepackage{appendix}

\usepackage[noend]{algpseudocode}    
\usepackage{algorithm}  

\newcommand{\cmmnt}[1]{\ignorespaces}

\def\wh{\widehat}

\def\rank{\hbox{rank}}

\def\boxit#1{\vbox{\hrule\hbox{\vrule\kern6pt
          \vbox{\kern6pt#1\kern6pt}\kern6pt\vrule}\hrule}}

\def\cov{\hbox{cov}}

\def\trace{\mbox{trace}}

\def\bse{\begin{eqnarray*}}
\def\ese{\end{eqnarray*}}
\def\be{\begin{eqnarray}}
\def\ee{\end{eqnarray}}
\def\bq{\begin{equation}}
\def\eq{\end{equation}}
\def\bse{\begin{eqnarray*}}
\def\ese{\end{eqnarray*}}

\def\wh{\widehat}

\def\Proof{\textbf{Proof: }}

\def\btheta{{\boldsymbol \theta}}

\newtheorem{theorem}{Theorem}
\newtheorem{lemma}{Lemma}[section]
\newtheorem{rem}{Remark}[section]

\newcommand{\bI}{\mathbf{I}}

\newcommand{\bZ}{\mathbf{Z}}

\newcommand{\bC}{\mathbf{C}}
\newcommand{\bM}{\mathbf{M}}

\newcommand{\bL}{\mathbf{L}}
\newcommand{\bD}{\mathbf{D}}
\newcommand{\bE}{\mathbf{E}}
\newcommand{\bW}{\mathbf{W}}

\newcommand{\bz}{\mathbf{z}}
\newcommand{\bs}{\mathbf{s}}

\newcommand{\bv}{\mathbf{v}}

\newcommand{\bw}{\mathbf{w}}

\newtheorem{defi}{Definition}[section]

\def\supp{\mathop{\rm supp}\nolimits}

\def\cov{\mathop{\rm cov}\nolimits}

\def\diam{\mbox{diam}}
\def\dist{\mbox{dist}}

\def\trace{\mbox{trace}}

\newcommand{\mydet}[1]{\vert #1 \vert}

\DeclareMathAlphabet{\mathitbf}{OML}{cmm}{b}{it}
\def\thetab{\bm{\theta}}

\def\LL{\mathcal{L}}

\def\bh{\mathbf{h}}

\def\H{\mathcal{H}}

\def\RR{\mathbb{R}}

\def\Proof{\textbf{Proof: }}

\def\jmath{j}

\begin{document}

\begin{frontmatter}

\title{Likelihood Approximation With Hierarchical Matrices For Large Spatial Datasets}

\author[Alex]{Alexander Litvinenko\corref{cor1}}
\ead{alexander.litvinenko@kaust.edu.sa}
\ead[url]{bayescomp.kaust.edu.sa, sri-uq.kaust.edu.sa}
\author[Alex]{Ying Sun}
\ead{ying.sun@kaust.edu.sa}
\ead[url]{es.kaust.edu.sa}
\author[Alex]{Marc G. Genton}
\ead{marc.genton@kaust.edu.sa}
\ead[url]{stsda.kaust.edu.sa}
\author[Alex]{David E. Keyes}
\ead{david.keyes@kaust.edu.sa}
\ead[url]{ecrc.kaust.edu.sa}

\address[Alex]{King Abdullah University of Science and Technology (KAUST),
Thuwal 23955-6900, Saudi Arabia}


\cortext[cor1]{Corresponding author}

\begin{abstract}
We consider measurements to estimate the unknown parameters (variance, smoothness, and covariance length) of a covariance function by maximizing the joint Gaussian log-likelihood function. To overcome cubic complexity in the linear algebra, we approximate the discretized covariance function in the hierarchical ($\mathcal{H}$-) matrix format. The $\mathcal{H}$-matrix format has a log-linear computational cost and storage $\mathcal{O}(kn \log n)$, where the rank $k$ is a small integer, and $n$ is the number of locations.
The $\mathcal{H}$-matrix technique allows us to work with general covariance matrices efficiently, since $\mathcal{H}$-matrices can approximate inhomogeneous covariance functions, with a fairly general mesh that is not necessarily axes-parallel, and neither the covariance matrix itself nor its inverse has to be sparse. We research how the $\mathcal{H}$-matrix approximation error influences on the estimated parameters.
We demonstrate our method with Monte Carlo simulations with known true values of parameters and an application to soil moisture data with unknown parameters.
The C, C++ codes and data are freely available.
\end{abstract}


\end{frontmatter}

\section{Introduction}\label{sec:intro}
The number of measurements that must be processed for statistical modeling in environmental applications is usually very large, and these measurements may be located irregularly across a given geographical region. This makes the computing procedure expensive and the data difficult to manage.
These data are frequently modeled as a realization from a stationary Gaussian spatial random field. 
Specifically, we let $\bZ=\{Z(\bs_1),\ldots,Z(\bs_n)\}^\top$, where $Z(\bs)$ is a Gaussian random field indexed by a spatial
location $\bs \in \Bbb{R}^d$, $d\geq 1$. Then, we assume that $\bZ$ has mean zero and a stationary parametric covariance function 
$C(\bh;\btheta)=\cov\{Z(\bs),Z(\bs+\bh)\}$, where $\bh\in\Bbb{R}^d$ is a spatial lag vector and $\btheta\in\Bbb{R}^q$ is
the unknown parameter vector of interest. Statistical inferences about $\btheta$ are often based on the Gaussian
log-likelihood function:
\begin{equation}
\label{eq:likeli}
\LL(\thetab)=-\frac{n}{2}\log(2\pi) - \frac{1}{2}\log \mydet{\bC(\thetab)}-\frac{1}{2}\bZ^\top \bC(\thetab)^{-1}\bZ,
\end{equation}
where the covariance matrix $\bC(\thetab)$ has entries $C(\bs_i-\bs_j;\btheta)$, $i,j=1,\ldots,n$. The maximum likelihood
estimator of $\btheta$ is the value $\wh \btheta$ that maximizes (\ref{eq:likeli}). When the sample size $n$ is large,
the evaluation of (\ref{eq:likeli}) becomes challenging, due to the computation of the inverse and log-determinant of the $n$-by-$n$ covariance matrix
$\bC(\thetab)$. Indeed, this requires ${\cal O}(n^2)$ memory and ${\cal O}(n^3)$ computational steps. Hence, scalable methods that can process larger sample sizes are needed.

Stationary covariance functions, discretized on a rectangular grid, have block Toeplitz structure. This structure can be further extended to a block circulant form 
and resolved with the Fast Fourier Transform (FFT) \cite{WHITTLE54, DAHLHAUS87, guinness2017circulant, stroud2017bayesian, Dietrich2}. The computing cost, in this case, is $\mathcal{O}(n\log n)$. However, this approach either does not work for data measured at irregularly spaced locations or requires expensive, non-trivial modifications.

During the past two decades, a large amount of research has been devoted to tackling the aforementioned computational challenge of developing scalable methods: for example,  low-rank tensor methods \cite{litv17Tensor, nowak2013kriging},
covariance tapering \cmmnt{\citep}\cite{Furrer2006,Kaufman2008, sang2012full}, likelihood approximations in both the spatial \cmmnt{\citep}\cite{Stein:Chi:Wetly:2004, Stein2013} and spectral  \cmmnt{\citep}\cite{Fuentes:2007} domains, 
latent processes such as Gaussian predictive processes \cmmnt{\citep}\cite{Banerjee:Gelfand:Finley:Sang:2008} and fixed-rank kriging
\cmmnt{\citep}\cite{Cressie:Johannesson:2008}, and Gaussian Markov random-field approximations \cmmnt{\citep}\cite{rue2002fitting,rue2005gaussian, fuglstad2015does}; see \cmmnt{\citet}\cite{Sun2012} for a review. 
A generalization of the Vecchia approach and a general Vecchia framework was introduced in \cite{Katzfuss17, Vecchia88}.
Each of these methods has its strengths and drawbacks. For instance, covariance tapering sometimes performs even worse than assuming independent blocks in the covariance  \cmmnt{\citet}\cite{stein2013statistical}; low-rank approximations have their own limitations \cmmnt{\citet}\cite{stein2014limitations}; and Markov models  depend on the observation locations, requiring irregular locations to be realigned on a much finer grid with missing values~\cmmnt{\citep}\cite{sun2016statistically}. A matrix-free approach for solving the multi-parametric Gaussian maximum likelihood problem was developed in \cite{Anitescu2012}. To further improve on these issues,
other methods that have been recently developed include the nearest-neighbor Gaussian process models \cmmnt{\citep}\cite{Datta:Banerjee:Finley:Gelfand:2015}, low-rank update \cite{SaibabaKitanidis15UQ}, multiresolution Gaussian process models \cmmnt{\citep}\cite{nychka2015multiresolution}, equivalent kriging \cmmnt{\citep}\cite{kleiber2015equivalent}, 
multi-level restricted Gaussian maximum likelihood estimators \cmmnt{\citep}\cite{Castrillon16}, and hierarchical low-rank approximations \cmmnt{\citep}\cite{Huang2016}. Bayesian approaches to identify unknown or uncertain parameters could be also applied \cite{rosic2013parameter, rosic2012sampling, matthies2012parametric, litvinenko2013inverse, matthies2012parametric, pajonk2012deterministic}. 

In this paper, we propose using the so-called hierarchical ($\H$-) matrices for approximating dense matrices with numerical complexity and storage $\mathcal{O}(k^{\alpha}n \log^{\alpha} n)$, where $n$ is the number of measurements; $k\ll n$ is the rank of the hierarchical matrix, which defines the quality of the approximation; and $\alpha=1$ or 2. $\H$-matrices are suitable for general grids and are also working well for large condition numbers. Previous results \cite{HackHMEng} show that the $\H$-matrix technique is very stable when approximating the matrix itself \cite{Li14, saibaba2012application, BallaniKressner, harbrecht2015efficient, ambikasaran2013large, BoermGarcke2007, SaibabaKitanidis12}, its inverse \cite{Ambikasaran16, ambikasaran2013large, bebendorf2003existence}, its Cholesky decomposition \cite{BebendorfSpEq16, bebendorf2007Why}, and the Schur complement (i.e., the conditional covariance matrix) \cite{HackHMEng, MYPHD, litvinenko2017partial}. The $\H$-matrix technique for inference including parameter estimation and MLE in Gaussian process regression has also been
proposed in \cite{ambikasaran2013large, Ambikasaran16}, and uncertainty quantification and conditional realizations in \cite{SaibabaKitanidis12}.
%
\begin{figure}
 \centering
 \includegraphics[width=3.5in]{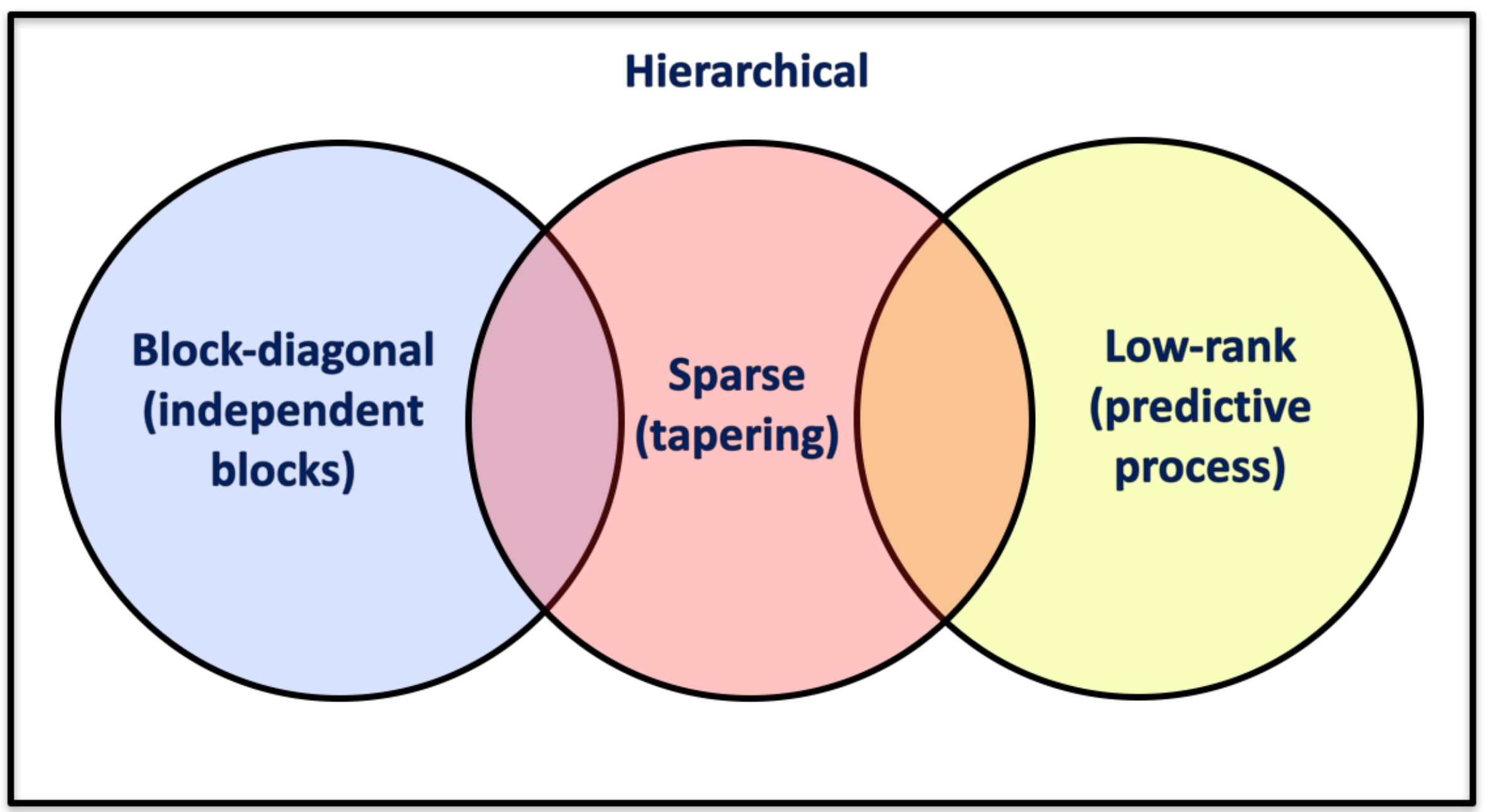}
 \caption{Scheme of approximations to the underlying dense covariance matrix (and corresponding methods). Hierarchically low-rank matrices generalize purely structure-based or global schemes with analysis-based locally adaptive schemes.}
\label{fig:Hscheme}
 \end{figure}

Other motivating factors for applying the $\H$-matrix technique include the following:
\begin{enumerate}
\item $\H$-matrices are more general than other compressed matrix representations (see scheme in Fig.~\ref{fig:Hscheme});
\item The $\H$-matrix technique allows us to compute not only the matrix-vector products (e.g., like fast multipole methods), but also the more general classes of functions, such as $\bC(\thetab)^{-1}$, $\bC(\thetab)^{1/2}$, $\mydet{\bC(\thetab)}$, $\exp\{\bC(\thetab)\}$, resolvents, Cholesky decomposition and, many others \cmmnt{\citep}\cite{HackHMEng};
\item The $\H$-matrix technique is well-studied and has a solid theory, many examples, and multiple sequential and parallel implementations;
\item The $\H$-matrix accuracy is controllable by the rank, $k$, or by the accuracy $\varepsilon$. The full rank gives an exact representation;
\item The $\H$-matrix technique preserves the structure of the matrix after 
the Cholesky decomposition and the inverse have been computed (see Fig.~\ref{fig:Hexample3});
\item There are efficient rank-truncation techniques to keep the rank small after the matrix operations; for instance, the Schur complement and matrix inverse can be approximated again in the $\H$-matrix format.
\end{enumerate}

Figure~\ref{fig:C} shows an $\mathcal{H}$-matrix approximation of the covariance matrix from a discretized ($n=16,641$) exponential covariance function on the unit square, its Cholesky approximation~\ref{fig:L}, and its $\H$-matrix inverse~\ref{fig:iC}.
The dark (or red) blocks indicate the dense matrices and the grey (green) blocks indicate the rank-$k$ matrices; the number inside each block is its rank. The steps inside the blocks show the decay of the singular values in $\log$ scale. The white blocks are empty.
\begin{figure}[htbp!]
 \centering
    \begin{subfigure}[b]{0.32\textwidth}
     \centering
      \caption{}
       \includegraphics[width=5cm]{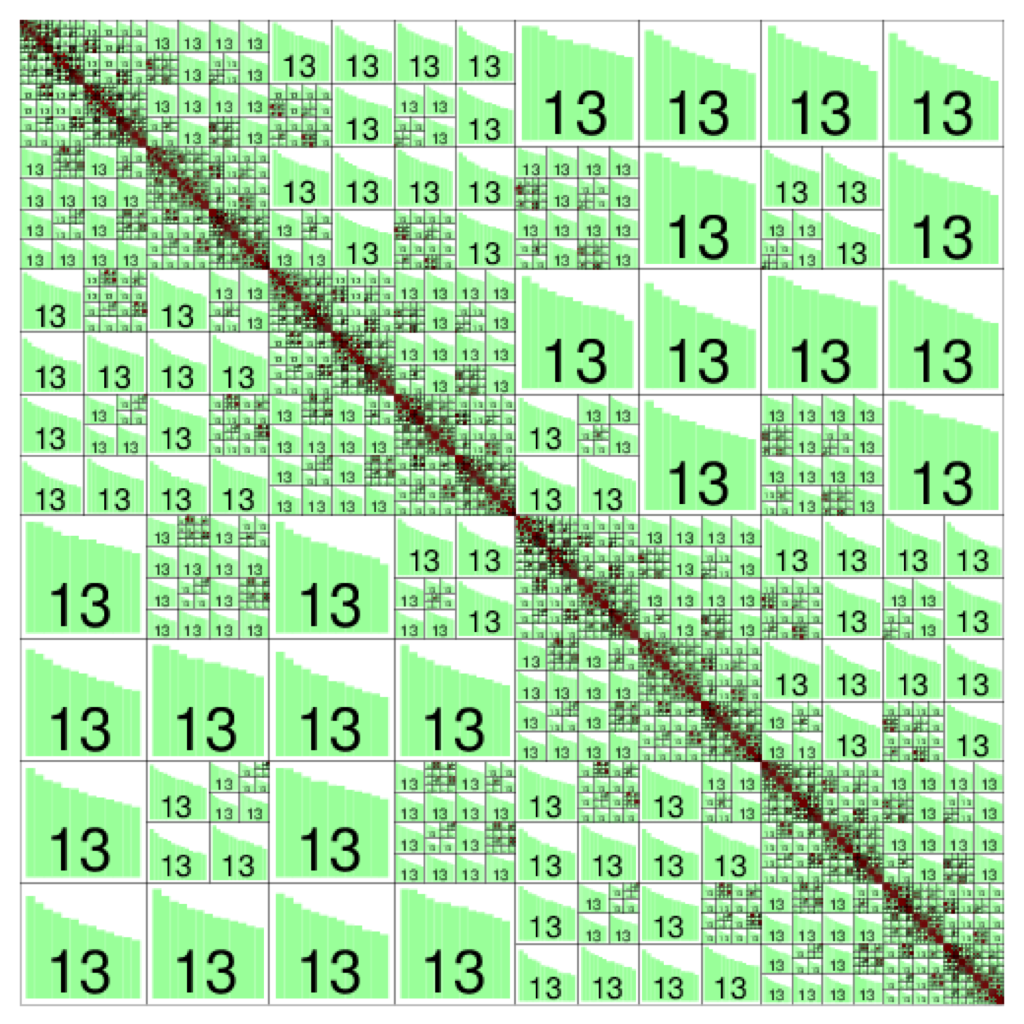}\vspace{-.1cm}
        \label{fig:C}
    \end{subfigure}
   \begin{subfigure}[b]{0.32\textwidth}
     \centering
        \caption{}
 \includegraphics[width=5cm]{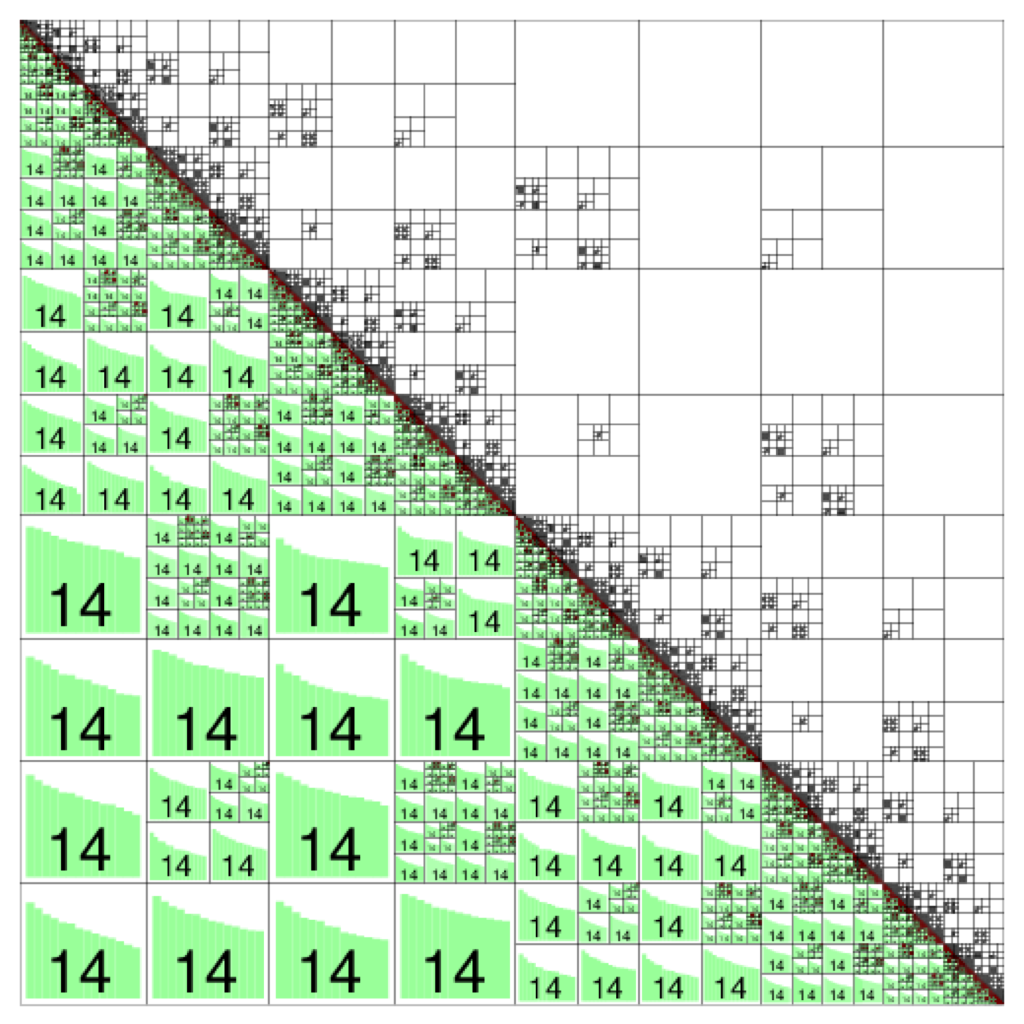}\vspace{-.1cm}
        \label{fig:L}
    \end{subfigure}
   \begin{subfigure}[b]{0.32\textwidth}
     \centering
       \caption{}
 \includegraphics[width=5cm]{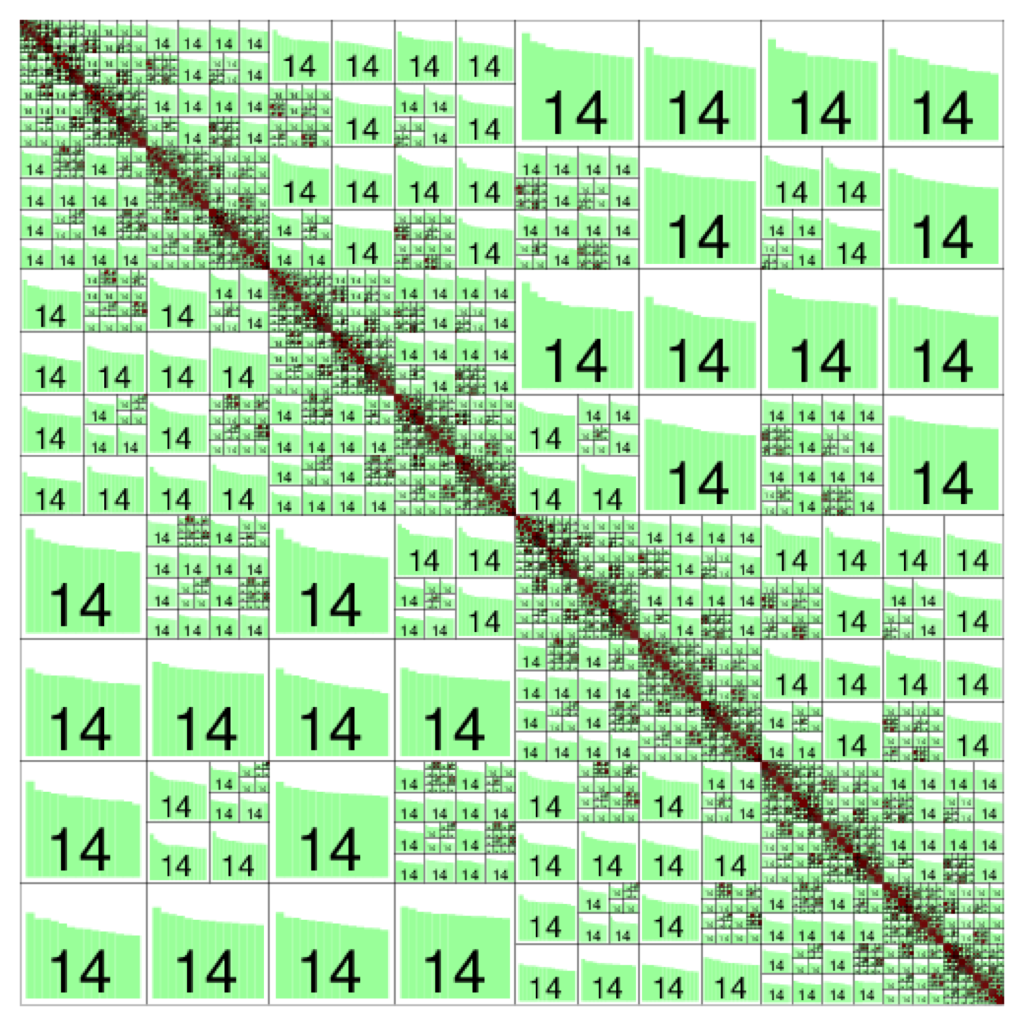}\vspace{-.1cm}
        \label{fig:iC}
    \end{subfigure}
 \caption{(a) The $\mathcal{H}$-matrix approximation of an $n\times n$ covariance matrix from a discretized exponential covariance function on the unit square, with $n=16,641$, unit variance, and length scale of 0.1. The dimensions of the densest (dark) blocks are $32\times 32$ and the maximal rank is $k=13$. (b) An example of the corresponding Cholesky factor with maximal rank $k=14$. (c) The inverse of the exponential covariance matrix (precision matrix) with maximal rank $k=14$.}
 \label{fig:Hexample3}
 \end{figure}
 %
%
%
%
%
In the last few years, there has been great interest in numerical methods for representing and approximating large covariance matrices in the applied mathematics community \cmmnt{\citep}\cite{Rasmussen05, BoermGarcke2007,saibaba2012application,nowak2013kriging,ambikasaran2013large,ambikasaran2014fast,si2014memoryef,BallaniKressner}. 

Recently, the maximum likelihood estimator for parameter-fitting Gaussian observations with a Mat\'ern covariance matrix was computed via a framework for unstructured observations in two spatial dimensions, which allowed the evaluation of the log-likelihood and its gradient with computational complexity $\mathcal{O}(n^{3/2})$; the method relied on the recursive skeletonization factorization procedure \cite{ho2015hierarchical, martinsson2005fast}. However, the consequences of the approximation on the maximum likelihood estimator were not studied.

In \cite{BoermGarcke2007}, the authors computed the exact solution of a Gaussian process regression by replacing the kernel matrix
with a data-sparse approximation, which they called the $\H^2$-matrix technique, cf. \cite{Li14} . 
It is more complicated than $\H$-matrix technique, but has computational complexity and storage cost of  $\mathcal{O}(kn)$.

\blue{The same $\H^2$-matrix technique for solving large-scale 
stochastic linear inverse problems with applications in subsurface modeling was demonstrated in \cite{ambikasaran2013large}. 
The authors explored the sparsity
of the underlying measurement operator, demonstrated the effectiveness by solving a realistic crosswell
tomography problem, quantified the uncertainty in the solution and provided an optimal capture
geometry for this problem. Their algorithm
was implemented in C++ and is available online.}

The authors of \cite{si2014memoryef} observed that the structure of shift-invariant kernels changed from low-rank to block-diagonal (without any low-rank structure) when they varied the scale parameter. Based on this observation, they proposed a new kernel approximation algorithm, which they called the Memory-Efficient Kernel Approximation. That approximation considered both the low-rank and the clustering structure of the kernel matrix. They also showed that the resulting algorithm outperformed state-of-the-art low-rank kernel approximation methods regarding speed, approximation error, and memory usage.
The BigQUIC method for a sparse inverse covariance estimation of a million variables was introduced in \cite{QUIC13}.
This method could solve $\ell_1$-regularized Gaussian Maximum Likelihood Estimation (MLE-) problems with dimensions of one-million.
In \cite{BallaniKressner}, the authors estimated the covariance matrix of a set of normally distributed random vectors. To overcome large numerical issues in the high-dimensional regime, they computed the (dense) inverses of sparse covariance matrices using $\H$-matrices. This explicit representation enables them to ensure positive definiteness of each Newton-like iterate in the resulting convex optimization problem. In conclusion, the authors compare their new $\H$-QUIC method with the existing BIG-QUIC method~\cite{QUIC13}.
In \cite{harbrecht2015efficient}, the authors offered to use $\H$-matrices for the approximation of random fields. In particular, they approximated Mat\'ern covariance matrices in the $\H$-matrix format, suggested the pivoted $\H$-Cholesky method and provided an a posteriori error estimate in the trace norm.
In \cite{saibaba2012application, SaibabaKitanidis12}, the authors applied $\H$-matrices to linear inverse problems in large-scale geostatistics.
\blue{They started with a detailed explanation of the $\H$-matrix technique, then reduced the cost of dense matrix-vector products, combined $\H$-matrices with a matrix-free Krylov subspace and solved the system of equations that arise from
the geostatistical approach. They illustrated the performance of their algorithm on an application, for
monitoring CO$_2$ concentrations using crosswell seismic tomography. The largest problem size was $n=10^6$. The code is available online.}

\blue{In \cite{ambikasaran2014fast}, the authors considered the new matrix factorization for Hierarchical Off-Diagonal Low-Rank (HODLR) matrices. A HODLR matrix is an $\H$-matrix with the weak admissibility condition (see Appendix~\ref{sec:adm}). All off-diagonal sub-blocks of a HODLR matrix are low-rank matrices, and this fact significantly simplifies many algorithms. 
The authors showed that typical covariance functions can be hierarchically factored into a product of block low-rank updates of the identity matrix, yielding an $\mathcal{O}(n \log^2 n)$ algorithm for inversion, and used this product for evaluation of the determinant with the cost $\mathcal{O}(n \log n)$, with further direct calculation of probabilities in high dimensions. Additionally, they used this HODLR factorization to speed up prediction and to infer unknown hyper-parameters. The provided formulas for the factorization and the determinant hold only for HODLR matrices. How to extend this case to general $\H$-matrices, where off-diagonal block are allowed to be also $\H$-matrices, is not so clear. For instance, conditional covariance matrices may easily lose the HODLR structure. The largest demonstrated problem size was $n=10^6$. The authors made their code freely available on GitHub.}

\textcolor{black}{Summarizing the literature review, we conclude that $\H$-matrices are already known to the statistical community. It is well-known that Mat\'ern covariance functions could be approximated in the $\H$-matrix format ($\H$, $\H^2$, HODLR). In \cite{BallaniKressner}, the authors used $\H$-matrices to compute derivatives of the likelihood function. Approximation of the inverse and Cholesky factorization is working in practice, but requires more theoretical analysis for error estimations (the existing theory is mostly developed for elliptic PDEs and integral equations). The influence of the $\H$-matrix approximation error on the quality of the estimated unknown parameters is not well studied yet. To research this influence with general $\H$-matrices (and not only with the simple HODLR format) is the main contribution of this work.}

The rest of this paper is structured as follows. In Section~\ref{sec:Hcov}, we remind the $\H$-matrix technique and review the $\H$-matrix approximation of Mat\'ern covariance functions. 
Section~\ref{sec:HAGL} contains the hierarchical approximation of the Gaussian likelihood function, the algorithm for parameter estimation, and a study of the theoretical properties and complexity of the approximate log-likelihood \cite{Ipsen15, Ipsen05}. Results from Monte Carlo simulations, where the true parameters are known, are reported in Section~\ref{sec:MC}.
An application of the hierarchical log-likelihood methodology to soil moisture data, where parameters are unknown, is considered in
Section~\ref{sec:moisture}. We end the paper with a discussion in Section~\ref{sec:Conclusion}. The derivations of the theoretical results are provided in the Appendix~\ref{appendix:A}, more details about $\H$-matrices in Appendix~\ref{sec:adm}.


\section{Hierarchical Approximation of Covariance Matrices}
\label{sec:Hcov}
Hierarchical matrices have been described in detail \cite{HackHMEng, Part1, GH03, weak, MYPHD}. Applications of the $\H$-matrix technique to the covariance matrices can be found in \cite{Ambikasaran16, SaibabaKitanidis12, saibaba2012application, BoermGarcke2007, harbrecht2015efficient, ambikasaran2013large, ambikasaran2014fast, khoromskij2009application}. There are many implementations exist. To our best knowledge, the HLIB library\footnote{http://www.hlib.org/} is not supported anymore, but it could be used very well for academic purposes, the $\H^2$-library\footnote{https://github.com/H2Lib, developed by Steffen Boerm and his group, Kiel, Germany} operates with $\H^2$-matrices, is actively supported and contains some new idea like parallel implementation on GPU. The HLIBPro library\footnote{https://www.hlibpro.com/, developed by R. Kriemann, Leipzig, Germany} is actively supported commercial, robust, parallel, very tuned, well tested, but not open source library. It contains about 10 open source examples, which can be modified for the user's personal needs. There are some other implementations (e.g., from M. Bebendorf or Matlab implementations), but we do not have experience with them.
To new users, we recommend to start with free open source HLIB or H2Lib libraries and then to move to HLIBPro, which is free for academic purposes.\\ 

\subsection{Hierarchical matrices}
\label{sec:Happrox}

In this section, we review the definition of $\H$-matrices and show how to approximate covariance matrices using the $\H$-matrix format. The $\H$-matrix technique is defined as a hierarchical partitioning of a given matrix
into sub-blocks followed by the further approximation of the majority of these sub-blocks by low-rank matrices (see Fig.~\ref{fig:Hexample3}).
To define which sub-blocks can be approximated well by low-rank matrices and which cannot, a so-called admissibility condition is used (see more details in Appendix~\ref{sec:adm}). There are different admissibility conditions possible: weak, strong, domain decomposition  based, and some others (Fig.~\ref{fig:Hexample_adm}). The admissibility condition may depend on parameters (e.g., the covariance length).

To define an $\H$-matrix some other notions are needed: index set $I$, clusters (denoted by $\tau$ and $\sigma$), cluster tree $T_{I}$, block cluster tree $T_{I\times I}$, admissibility condition (see Appendix~\ref{sec:adm}). 
We start with the index set $I=\{0,\ldots,n-1\}$, which corresponds to the available measurements in $n$ locations.
After the hierarchical decomposition of the index set $I$ into sub-index sets has been completed (or in other words, a cluster tree $T_I$ is constructed), 
the block cluster tree (denoted by $T_{I\times I}$, see more details in Appendix~\ref{sec:adm}) together with the admissibility condition decides which sub-blocks can be approximated by low-rank matrices. For definitions and example of cluster trees and corresponding block cluster trees (block partitioning) see Appendix~\ref{sec:adm} and Fig.~\ref{fig:ct_bct}.

On the first step, the matrix is divided into four sub-blocks. The hierarchical tree $T_{I \times I}$ tells how to divide. Then each (or some) sub-block(s) is (are) divided again and again until sub-blocks are sufficiently small. The resulting division is hierarchical. The procedure stops when either one of the sub-block sizes is $n_{\mbox{min}}$ or smaller ($n_{\mbox{min}}\leq 128$), or when this sub-block can be approximated by a low-rank matrix.

Another important question is how 
to compute these low-rank approximations. \textcolor{black}{The HLIB library uses a well-known method, called Adaptive Cross Approximation (ACA)} algorithm \cite{TyrtyshACA,ACA,Winter}, which performs the approximations with a linear complexity $\mathcal{O}(kn)$ in contrast to $\mathcal{O}(n^3)$ by SVD.

%
%
%
%
%
%
\begin{rem} 
Errors in the $\H$-matrix approximation may destroy the symmetry of the symmetric positive definite covariance matrix, causing the symmetric blocks to have different ranks. As a result, the standard algorithms used to compute the Cholesky decomposition may fail. A remedy to this is defining $\bC:=\frac{1}{2}(\bC+\bC^\top)$.
\end{rem}
\begin{rem} 
\textcolor{black}{Errors in the $\H$-matrix approximation may destroy the positive definiteness of the covariance matrix. Especially for matrices which have very small (very close to zero) eigenvalues. 
A remedy is: 1) to use a more robust, e.g., block $\H$-Cholesky, algorithm; 2) use $\bL\bD\bL^\top$ factorization instead of $\bL\bL^\top$ or 3) add a positive diagonal $\tau^2 \cdot\bI$ to $\bC$.}
\end{rem}
\begin{rem}
\textcolor{black}{We use both notations $\bC\in \mathbb{R}^{n \times n}$ and $\bC\in \mathbb{R}^{I \times I}$. In this work it is the same, since $\vert I \vert = n$.
The notation $\bC\in \mathbb{R}^{I \times I}$ is useful when it is important to differentiate between $\bC\in \mathbb{R}^{I \times I}$ and $\bC\in \mathbb{R}^{J \times J}$, where $I$ and $J$ are two different index sets of the same size.}
\end{rem}
To define the class of $\H$-matrices, we assume that the cluster tree $T_{I}$ and block cluster tree $T_{I\times I}$ are already constructed.
\begin{defi}
\label{def:Hmatrix}
We let $I$ be an index set and $T_{I\times I}$ a hierarchical division of the index set product, $I\times I$, into sub-blocks. The set of $\H$-matrices with the maximal sub-block rank $k$ is
\begin{equation*}
\mathcal{H}(T_{I\times I},k):=\{\bC \in \mathbb{R}^{I \times I}\, \vert \, \mbox{rank}(\bC\vert_b) \leq k \, \text{  for all admissible blocks  } b \text{  of } T_{I\times I}\},
\end{equation*}
where $k$ is the maximum rank.
Here, $\bC\vert_b = (c_{ij})_{(i,j)\in b}$ denotes the matrix block of $\bC = (c_{ij})_{i,j\in I}$ corresponding to the sub-block $b \in T_{I\times I}$.
\end{defi} 
Blocks that satisfy the admissibility condition can be approximated by low-rank matrices; see  \cite{Part1}.
An $\H$-matrix approximation of $\bC$ is denoted by $\widetilde{\bC}$.
%
The storage requirement of $\widetilde \bC$ and the matrix vector multiplication cost $\mathcal{O}(kn\log n)$, the matrix-matrix addition costs $\mathcal{O}(k^2n\log n)$, and the matrix-matrix product and the matrix inverse cost $\mathcal{O}(k^2n\log^2 n)$; see \cite{Part1}. 
%
%
%
\begin{rem} 
\label{rem:adap_fixed}
There are two different  $\H$-matrix approximation strategies, \cite{Winter,HackHMEng}. 1) \textit{The fixed rank strategy} (Fig.~\ref{fig:Hexample3}), i.e., each sub-block has a maximal rank $k$. It could be sufficient or not, but it is simplify theoretical estimates for the computing time and storage. In this case we write $\widetilde{\bC} \in \mathcal{H}(T_{I\times I};k)$  2) \textit{The adaptive rank strategy} (Fig.~\ref{fig:Hexample_adm}, left sub-block), i.e., where absolute accuracy (in the spectral norm) for each sub-block $\bM$ is $\varepsilon$ or better (smaller):
\begin{equation*}
k:=\min\{\tilde{k}\in \mathbb{N}_0 \vert \,\exists\, \tilde{\bM}\in\mathcal{R}(\tilde{k},n,m):\Vert \bM-\tilde{\bM} \Vert \leq \varepsilon \Vert \bM\Vert\},
\end{equation*}
where $\mathcal{R}(\tilde{k},n,m):=\{\bM\in \mathbb{R}^{n\times m} \vert \rank(\bM)\leq \tilde{k}\}$ is a set of low-rank matrices of size $n\times m$ of the maximal rank $\tilde{k}$.
The second strategy is better, since it allows to have an adaptive rank for each sub-block, but it makes it difficult to estimate the total computing cost and storage.
In this case, we write $\widetilde{\bC} \in \mathcal{H}(T_{I\times I};\varepsilon)$.\\ 
\end{rem}
The fixed rank strategy is useful for a priori evaluations of the computational resources and storage memory. The adaptive rank strategy is preferable for practical approximations and is useful when the accuracy in each sub-block is crucial. 
\textcolor{black}{Similarly, we introduce $\widetilde \LL(\thetab; k)$ and $\widetilde \LL(\thetab; \varepsilon)$ as $\H$-matrix approximations of the log-likelihood $\LL$ from (\ref{eq:likeli}).}
Sometimes we skip $k$ (or $\varepsilon$) and write $\widetilde \LL(\thetab)$.

\textcolor{black}{Thus, we conclude that different $\H$-matrix formats exist. For example, the weak admissible matrices (HODLR) result in a simpler block structure (Fig.~\ref{fig:Hexample_adm}, right), but the $\H$-matrix ranks could be larger than the ranks by standard admissible matrices. In \cite{weak}, the authors observed a factor of $\approx 3$ between the weak and standard admissible sub-blocks. In \cite{MYPHD}, the author observed that the HODLR matrix format might fail or result in very large ranks for 2D/3D computational domains or when computing the Schur complement. The $\H^2$ matrix format allows to get rid of the $\log$ factor in the complexity and storage but it is more complicated for understanding and implementation. We note that there are conversions from $\H$ to $\H^2$ matrix formats implemented (see HLIB and H2LIB libraries).}


\subsection{Mat\'{e}rn covariance functions}
\label{sec:Matern}

Among the many covariance models available, the Mat\'{e}rn family \cmmnt{\citep}\cite{Matern1986a} has gained widespread interest
in recent years. The Mat\'{e}rn form of spatial correlations was introduced into statistics as a flexible parametric class, with one parameter determining the smoothness of the underlying spatial random field \cite{Handcock1993a}. The varied history of this family of models can be found in \cite{Guttorp2006a}.

The Mat\'{e}rn covariance depends only on the distance $h:=\Vert \bs-\bs'\Vert $, where $\bs$ and $\bs'$ are any two spatial locations.
The Mat\'{e}rn class of covariance functions is defined as
\begin{equation}
\label{eq:MaternCov}
C(h;\btheta)=\frac{\sigma^2}{2^{\nu-1}\Gamma(\nu)}\left(\frac{h}{\ell}\right)^\nu K_\nu\left(\frac{h}{\ell}\right),
\end{equation}
%
with $\btheta=(\sigma^2,\ell,\nu)^\top$; where $\sigma^2$ is the variance; $\nu>0$ controls the smoothness of the random field, with larger values of $\nu$ corresponding to smoother fields; and $\ell>0$ is a spatial range parameter that measures how quickly the correlation of the random 
field decays with distance, with larger $\ell$ corresponding to a faster decay (keeping $\nu$ fixed) \cite{harbrecht2015efficient}. Here ${\cal K}_\nu$ denotes the modified Bessel function of the second kind of order $\nu$. When $\nu=1/2$, the Mat\'{e}rn covariance function reduces to the exponential covariance model and describes a rough field.
The value $\nu=\infty$ corresponds to a Gaussian covariance model that describes a very smooth, infinitely differentiable field.
Random fields with a Mat\'{e}rn covariance function are $\lfloor \nu-1 \rfloor$ times mean square differentiable. 


\subsection{Mat\'ern covariance matrix approximation}

By definition, covariance matrices are symmetric and positive semi-definite. The decay of the eigenvalues (as well as  the $\H$-matrix approximation) depends on the type of the covariance matrix and its smoothness, covariance length, 
computational geometry, and dimensionality. In this section, we perform numerical experiments with $\H$-matrix approximations.

To underline the fact that $\H$-matrix approximations can be applied to irregular sets of locations, we use irregularly spaced locations in the unit square (only for Tables~\ref{table:eps_det}, \ref{table:approx_compare_rank}):
\begin{equation}
\label{eq:mesh_pert}
\frac{1}{\sqrt{n}}\left( i-0.5+X_{i j}, \,\, j-0.5+Y_{i j}\right), \quad \text{for}\,\, i,j \in \{1,2,\ldots,\sqrt{n}\},
\end{equation}
where $X_{i j}$, $Y_{i j}$ are i.i.d. uniform on $(-0.4,0.4)$ for a total of $n$ observations;
see \cite{sun2016statistically}. The observations are ordered lexicographically by the ordered pairs $(i,j)$.

All of the numerical experiments herein are performed on a Dell workstation with 20 processors (40 cores) and total 128 GB RAM.
The parallel $\H$-matrix library, HLIBPro \cmmnt{\citep}\cite{HLIBPRO}, was used to build the Mat\'{e}rn covariance matrix, compute the Cholesky factorization, solve a linear system, calculate the determinant and the quadratic form. HLIBPro is fast and efficient; see the theoretical parallel scalability in Table \ref{tab:1}. Here $V(T)$ denotes the set of vertices, $L(T)$ the set of leaves in the block-cluster tree $T=T_{I\times I}$, and $n_{min}$  the size of a block when we stop further division into sub-blocks (see Section \ref{sec:Happrox}). Usually $n_{min}=\{32, 64, 128\}$, since a very deep hierarchy slows down computations.
\begin{table}[h!]
\begin{center}
\caption{Sequential and parallel complexity of the main linear algebra operations on $p$ processors.}\label{tab:1}
\begin{tabular}{|l|l|l|} 
\hline 
 Operation & Sequential Complexity & Parallel Complexity \cmmnt{\citep}\cite{HLIBPRO} (Shared Memory)\\ 
 \hline\hline
  build $\widetilde{\bC}$ & $\mathcal{O}(n\log n)$ & $\frac{\mathcal{O}(n\log n)}{p}+\mathcal{O}(\vert V(T) \setminus{L}(T) \vert)$\\ \hline
  store $\widetilde{\bC}$ & $\mathcal{O}(kn\log n)$ & $\mathcal{O}(kn\log n)$\\ \hline
  $\widetilde{\bC}\cdot \bz$ & $\mathcal{O}(kn\log n)$ & $\frac{\mathcal{O}(kn\log n)}{p} + \frac{n}{\sqrt{p}}$\\ \hline
  $\widetilde{\bC}^{-1}$ & $\mathcal{O}(k^2n\log^2 n)$ & $\frac{\mathcal{O}(n\log n)}{p}+\mathcal{O}(nn_{min}^2)$, $1\leq n_{min} \leq128$\\ \hline
  $\H$-Cholesky $\widetilde \bL$& $\mathcal{O}(k^2n\log^2 n)$ & $\frac{\mathcal{O}(n\log n)}{p}+\mathcal{O}(\frac{k^2n\log^2 n}{n^{1/d}})$\\ \hline
  $\mydet{\widetilde{\bC}}$ & $\mathcal{O}(k^2n\log^2 n)$ & $\frac{\mathcal{O}(n\log n)}{p}+\mathcal{O}(\frac{k^2n\log^2 n}{n^{1/d}})$, $d=1,2,3$\\ \hline
\end{tabular}
\end{center}
\end{table} 
%
%
%

Table \ref{table:eps_det}  shows the $\H$-matrix approximation errors for $\log \mydet{\widetilde{\bC}}$, $\widetilde{\bC}$ and the inverse $(\widetilde{\bL}\widetilde{\bL}^\top)^{-1}$, where $\widetilde{\bL}$ is the $\H$-Cholesky factor. The errors are computed in the Frobenius and spectral norms, where $\bC$ is the exact Mat\'{e}rn covariance matrix with $\nu=0.5$ and $\sigma^2=1$. The local accuracy in each sub-block is 
$\varepsilon$. The number of locations is $n=16{,}641$. The last column demonstrates the total compression ratio, c.r., which is equal to $1-$size($\widetilde{\bC}$)/size($\bC$). The exact values are $\log \mydet{\bC}=2.63$ for $\ell=0.0334$ and $\log \mydet{\bC}=6.36$ for $\ell=0.2337$. The uniformly distributed mesh points are taken in the unit square and perturbed as in (\ref{eq:mesh_pert}).
\begin{table}[h!]
\centering
\caption{The $\H$-matrix accuracy and compression rates (c.r.). Accuracy in each sub-block is 
$\varepsilon$; $n=16{,}641$, $\widetilde{\bC}$ is a Mat\'{e}rn covariance with $\nu=0.5$ and $\sigma^2=1$. The spatial domain is the unit square with locations irregularly spaced as in (\ref{eq:mesh_pert}). 
}
\begin{tabular}{|c|c|c|c|c|c|c|}
\hline
 $\varepsilon$ &{\scriptsize $\big|\log \mydet{\bC} -\log \mydet{\widetilde{\bC}}\big|$}  &$\big| \frac{\log \mydet{\bC} -\log \mydet{\widetilde{\bC}}}{\log \mydet{\widetilde{\bC}}}\big|$  & {\scriptsize $\Vert \bC-\widetilde{\bC} \Vert_F$} & $\frac{\Vert \bC-\widetilde{\bC} \Vert_2}{\Vert \widetilde{\bC}\Vert_2}$ & {\scriptsize $\Vert \bI-(\widetilde{\bL}\widetilde{\bL}^\top)^{-1}{\bC} \Vert_2$} & c.r. ($\%$)  \\
 \hline
 $\ell=0.0334$ &&&&&&\\
 \hline
 $10^{-1}$ & $3.2\cdot 10^{-4}$    & $1.2\cdot 10^{-4}$   & $7.0\cdot 10^{-3}$  & $7.6\cdot 10^{-3}$    & 2.9 & 91.8\\ 
 $10^{-2}$ & $1.6\cdot 10^{-6}$   & $6.0\cdot 10^{-7}$   & $1.0\cdot 10^{-3}$  & $6.7\cdot 10^{-4}$    & $9.9\cdot 10^{-2}$  & 91.6\\ 
 $10^{-4} $ & $1.8\cdot 10^{-9}$   & $7.0\cdot 10^{-10}$   & $1.0\cdot 10^{-5}$  & $7.3\cdot 10^{-6}$    & $2.0\cdot 10^{-3}$  & 89.8\\ 
 $10^{-8} $ & $4.7\cdot 10^{-13}$  & $1.8\cdot 10^{-13}$ & $1.3\cdot 10^{-9}$  & $6\cdot 10^{-10}$  & $2.1\cdot 10^{-7}$  & 87.3\\ \hline
 $\ell=0.2337$ &&&&&&\\ \hline
 $ 10^{-4}$ & $9.8\cdot 10^{-5}$    & $1.5\cdot 10^{-5}$   & $8.1\cdot 10^{-5}$  & $1.4\cdot 10^{-5}$    &$ 2.5\cdot 10^{-1}$  & 91.5\\ 
 $10^{-8}$  & $1.5\cdot 10^{-9}$  & $2.3\cdot 10^{-10}$ & $1.1\cdot 10^{-8}$  & $1.5\cdot 10^{-9}$    & $4\cdot 10^{-5}$     & 88.7\\ 
 \hline
\end{tabular}
\label{table:eps_det} 
\end{table}

\subsection{Convergence of the $\H$-matrix error vs. the rank $k$}
In Table~\ref{table:approx_compare_rank} we show the
dependence of the Kullback-Leibler divergence (KLD) and two matrix errors on the $\H$-matrix rank $k$ for the Mat\'{e}rn covariance function with parameters $\ell=\{0.25, 0.75\}$ and $\nu=1.5$, computed on the domain $\mathcal{G}=[0,1]^2$. \textcolor{black}{The KLD was computed as follow:
\begin{equation*}
D_{KL}(\bC,\widetilde{\bC}) = \frac{1}{2}\left \{\trace(\widetilde{\bC}^{-1}\bC) - n + \log\left (\frac{\det \widetilde{\bC}}{\det \bC}\right) \right \}.
\end{equation*}}
We can bound the relative error $\Vert \bC^{-1} - \widetilde{\bC}^{-1}\Vert/\Vert \bC^{-1} \Vert$ for the approximation of the inverse as
\begin{equation*}
\frac{\Vert \bC^{-1} - \widetilde{\bC}^{-1}\Vert}{\Vert \bC^{-1} \Vert}=\frac{\Vert (\bI-  \widetilde{\bC}^{-1}\bC)\bC^{-1}\Vert}{\Vert \bC^{-1} \Vert} \leq \Vert (\bI-  \widetilde{\bC}^{-1}\bC)\Vert.
\end{equation*} 
The spectral norm of $\Vert (\bI-  \widetilde{\bC}^{-1}\bC)\Vert$ can be estimated by few steps of the power iteration method (HLIB library uses 10 steps).
 The rank $k\leq 20$ is not sufficient to approximate
the inverse, and the resulting error $\Vert {\bC} (\widetilde{\bC})^{-1} -\bI \Vert_2 $ is large. One remedy could be to increase the rank, but it may not help for ill-conditioned matrices. The spectral norms of $\tilde{\bC}$ are $\Vert \tilde{\bC}_{(\ell=0.25)}\Vert_2=720$ and $\Vert \tilde{\bC}_{(\ell=0.75)}\Vert_2=1068$.

\textcolor{black}{
\begin{rem}
Very often, the nugget $\tau^2\bI$ is a part of the model or is just added to the main diagonal of $\bC$ to stabilize numerical calculations, i.e., $\widetilde{\bC}:=\widetilde{\bC}+\tau^2 \bI$,  \cite{LitvGentonSunKeyes17}. By adding a nugget, we ``push'' all the singular values away from zero. Adding a nugget effect to the list of unknown parameters $\btheta=(\ell,\nu, \sigma^2,\tau)^\top$ is straightforward. See the modified procedure in the GitHub repository \cite{LitvGitHubHcov}.
\end{rem}}
\begin{table}[h!]
\centering
\caption{Convergence of the $\H$-matrix approximation error vs. the $\H$-matrix rank $k$ of a Mat\'{e}rn covariance function with parameters $\ell=\{0.25, 0.75\}$, $\nu=1.5$, domain $\mathcal{G}=[0,1]^2$, $n=16{,}641$,}
\begin{tabular}{|c|cc|cc|cc|}
\hline
 $k$ & \multicolumn{2}{c|}{KLD}& \multicolumn{2}{c|}{$\Vert \bC - \widetilde{\bC} \Vert_2$} & \multicolumn{2}{c|}{$\Vert {\bC} (\widetilde{\bC})^{-1} -\bI \Vert_2  $} \\
     &  $\ell=0.25$   &  $\ell=0.75$ &  $\ell=0.25$   &  $\ell=0.75$ &  $\ell=0.25$   &  $\ell=0.75$ \\ 
\hline
 20  &  0.12  &  2.7  & 5.3e-7& 2e-7   & 4.5& 72\\ 
 30  &  3.2e-5&  0.4  & 1.3e-9& 5e-10  & 4.8e-3& 20\\ 
 40  &  6.5e-8& 1e-2  &1.5e-11& 8e-12  & 7.4e-6& 0.5\\ 
 50  &  8.3e-10& 3e-3 &2.0e-13& 1.5e-13&  1.5e-7& 0.1\\ 
\hline
\end{tabular}
\label{table:approx_compare_rank}  
\end{table}
\section{Hierarchical Approximation of Gaussian Likelihood}
\label{sec:HAGL}

\subsection{Parameter estimation}

We use the $\H$-matrix technique to approximate the Gaussian likelihood function.  The $\H$-matrix approximation of the exact log-likelihood $\LL(\thetab)$ defined in (\ref{eq:likeli}) is denoted by $\widetilde \LL(\thetab;k)$:
\begin{equation}
\label{eq:likeH}
\widetilde{\LL}(\thetab;k)=-\frac{n}{2}\log(2\pi) - \sum_{i=1}^n\log \{\widetilde{L}_{ii}(\thetab;k)\}-\frac{1}{2}\bv(\thetab)^\top \bv(\thetab),
 \end{equation}
 where  $\widetilde{\bL}(\thetab;k)$ is the rank-$k$ $\H$-matrix approximation of the 
Cholesky factor $\bL(\btheta)$ in $\bC(\btheta)=\bL(\btheta)\bL(\btheta)^\top$,
and $\bv(\thetab) $ is the solution of the linear system $\widetilde{\bL}(\thetab;k)\bv(\thetab)=\bZ$. Analogously, we can define  $\widetilde{\LL}(\thetab;\varepsilon)$.


To maximize $\widetilde{\LL}(\thetab;k)$ in (\ref{eq:likeH}), we use Brent's method \cmmnt{\citep}\cite{Brent73, BrentMethod07}, also known as Brent-Dekker\footnote{Implemented in GNU Scientific library, \url{https://www.gnu.org/software/gsl/} }. The  Brent-Dekker algorithm first uses the fast-converging secant method or inverse quadratic interpolation to maximize $\widetilde{\LL}(\thetab;\cdot)$. If those do not work, then it returns to the more robust bisection method.
We note that the maximization of the log-likelihood function is an ill-posed problem, since even very small perturbations in the covariance matrix $\bC(\btheta)$ may result in huge perturbations in the log-determinant and the log-likelihood. An alternative to $\bC(\btheta)=\bL(\btheta)\bL(\btheta)^\top$ is the $\bL(\btheta) \bD(\btheta) \bL^\top(\btheta)$ decomposition, which is more stable since it avoids extracting square roots of diagonal elements, i.e., $\bL \bD \bL^\top=(\bL\bD^{1/2})(\bL\bD^{1/2})^\top$. Very small negative diagonal elements can appear due to, e.g., the rounding off error.


\subsection{Computational complexity and accuracy}


We let $\bC(\btheta)\in \mathbb{R}^{n \times n}$ be approximated by an $\H$-matrix $\widetilde{\bC}(\btheta;k)$ with a maximal rank $k$. 
The $\H$-Cholesky decomposition of $\widetilde{\bC}(\btheta;k)$ costs $\mathcal{O}(k^2 n\log^2 n)$. The solution of
the linear system $\widetilde{\bL}(\thetab;k)\bv(\btheta)=\bZ$ costs $\mathcal{O}(k^2 n\log^2 n)$.
The log-determinant $\log \mydet{\widetilde \bC(\btheta;k)}=2\sum_{i=1}^n \log \{\widetilde L_{ii}(\btheta;k)\}$ is available for free. 
The cost of computing the log-likelihood function $\widetilde{\LL}(\thetab;k)$ is $\mathcal{O}(k^2 n\log^2 n)$ and the cost of computing the MLE  in $m$ iterations is ${\mathcal{O}(m k^2 n\log^2 n)}$.


Once we observe a realization $\bz$ of the random vector $\bZ$, we can quantify the accuracy of
the $\cal H$-matrix approximation of the log-likelihood function. 
Our main theoretical result is formulated below in Theorem~\ref{thm:MainLL}. 
%
\begin{theorem}(Accuracy of the hierarchical log-likelihood)\\
\label{thm:MainLL}
We let $\widetilde{\bC}(\btheta)$ be an $\H$-matrix approximation of the matrix $\bC(\btheta)\in \mathbb{R}^{n\times n}$, and $\bZ=\bz$ a data vector.
We let also the spectral radius $\rho(\widetilde{\bC}(\btheta)^{-1}{\bC}(\btheta)-\bI) <\varepsilon<1$. 
Then the following statements hold:
\begin{eqnarray*}
\big|  \log \mydet{\widetilde{\bC}(\btheta;k)} - \log \mydet{\bC(\btheta)} \big| &\leq& -n \log(1-\varepsilon)\approx n\varepsilon \quad \text{for small} \quad \varepsilon,\\
\vert \widetilde{\LL}(\btheta;k) - \LL(\btheta) \vert &\leq&  
\frac{1}{2}\cdot n \varepsilon + \frac{1}{2}\Vert \bz\Vert^2_2\cdot \Vert \widetilde{\bC}^{-1}\Vert_2  \cdot \varepsilon.
\end{eqnarray*}
\end{theorem}
%
\Proof See in the Appendix A. \\
The estimate in Remark~\ref{rem:NormZ} states how fast the probability decays while increasing the norm of $\bZ$. For simplicity, further on, we will assume $\Vert \bz \Vert_2 \leq c_0$.

In \cite{BallaniKressner}, the authors observed that the bound $n$ in the inequalities above is pessimistic and is hardly observed in numerical simulations. Though Theorem~\ref{thm:MainLL} is shown for the fixed rank arithmetics, it also holds for the adaptive rank arithmetics with $\widetilde{\bC}(\btheta;\varepsilon)$ and $\widetilde{\LL}(\btheta;\varepsilon)$.
\section{Monte Carlo Simulations}
\label{sec:MC}
We performed numerical experiments with simulated data to recover the true values of the parameters of the Mat\'ern covariance matrix, known to be $\thetab^*:=(\ell^*, \nu^*,\sigma^*)=(0.7, 0.9, 1.0)$.  In the first step, we constructed 50 independent data sets (replicates) of size $n\in \{ 128,64,\ldots,4,2\}\times 1{,}000 $, by multiplying the Cholesky factor  $\widetilde{\bL}(\thetab,10^{-10})$ on a Gaussian random vector $\bW \sim \mathcal{N}(\mathbf{0}, \bI)$, $\widetilde{\bC}(\thetab^*)=\widetilde{\bL}(\thetab^*)\widetilde{\bL}(\thetab^*)^\top$.
We took the locations (not the data) from the daily moisture example, which is described below in Sect.~\ref{sec:moisture}.  
After that we ran the optimization algorithm and tried to identify (recover) the ``unknown" parameters $(\ell^*, \nu^*,\sigma^*)$. The boxplots for each parameter over 50  replicates are plotted in Figure~\ref{fig:synthetic_boxes}

%
%
\begin{figure}[h!]
    \centering
    \begin{subfigure}[b]{0.48\textwidth}
     \centering
        \caption{}
       \includegraphics[width=8cm]{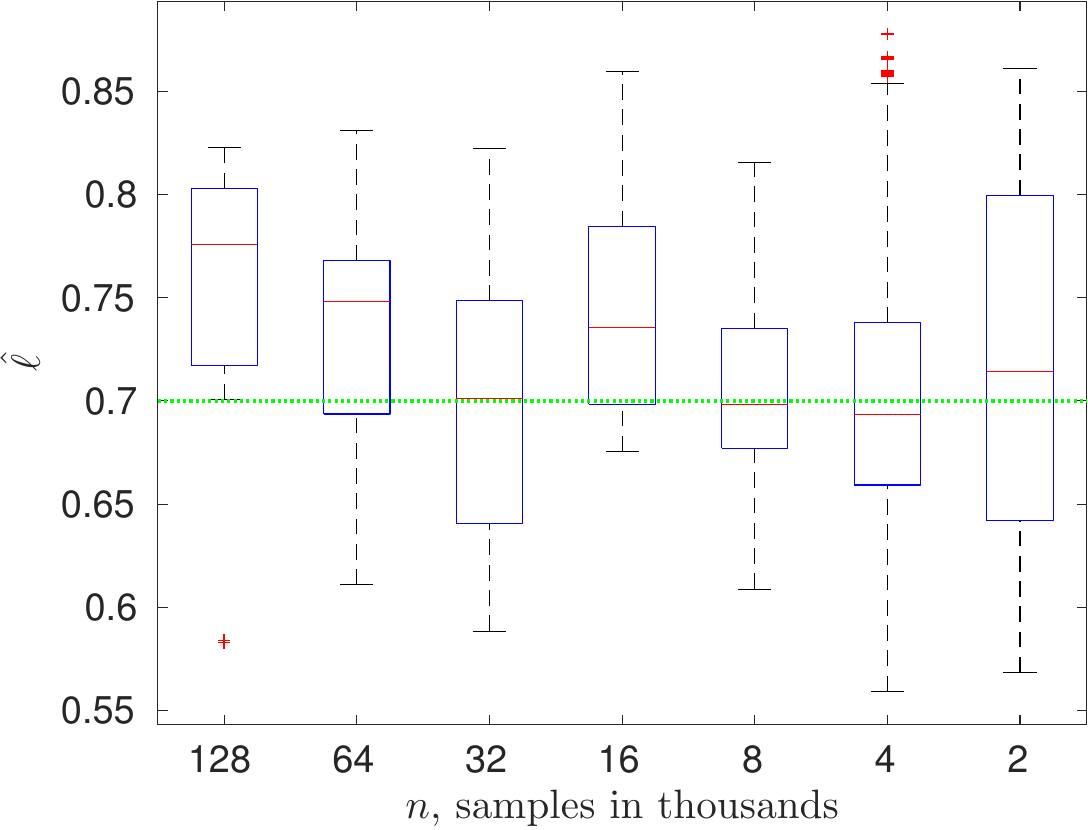}\vspace{-.1cm}
        \label{fig:ell_synt}
    \end{subfigure}
    ~ 
    \begin{subfigure}[b]{0.48\textwidth}
     \centering
        \caption{}
       \includegraphics[width=8cm]{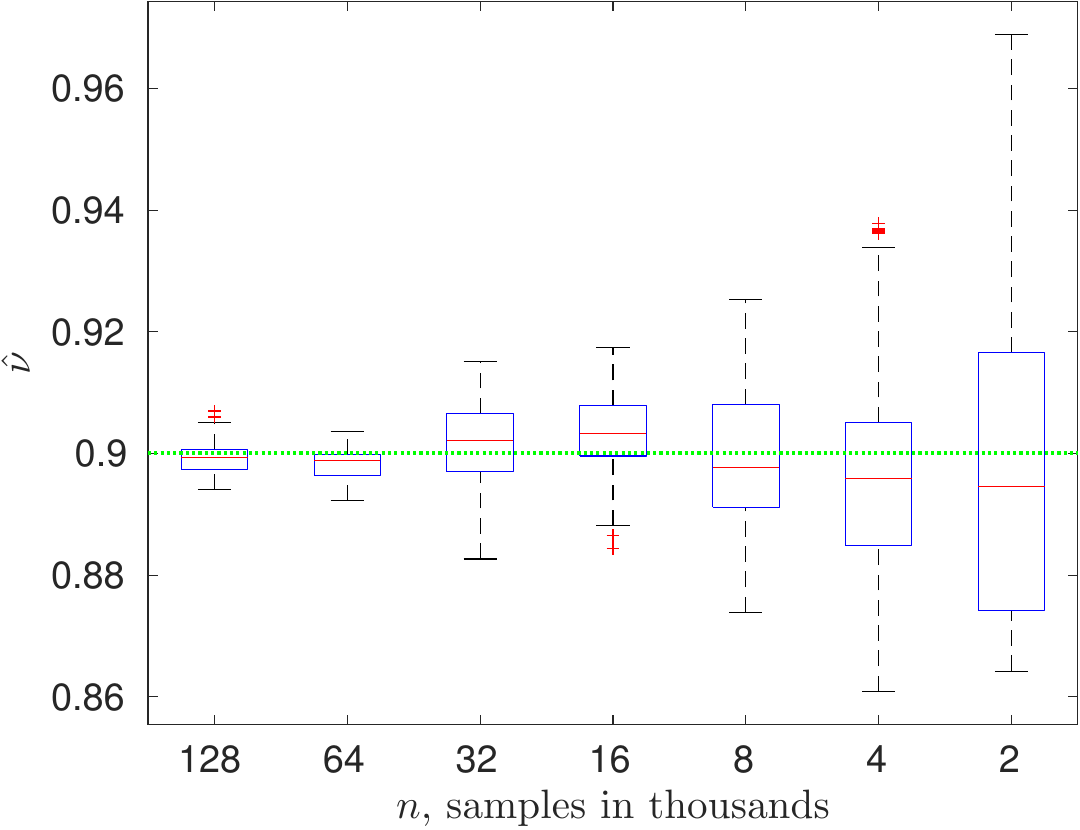}\vspace{-.1cm}
        \label{fig:nu_synt}
    \end{subfigure}\\
    \begin{subfigure}[b]{0.48\textwidth}
     \centering
        \caption{}
       \includegraphics[width=8cm]{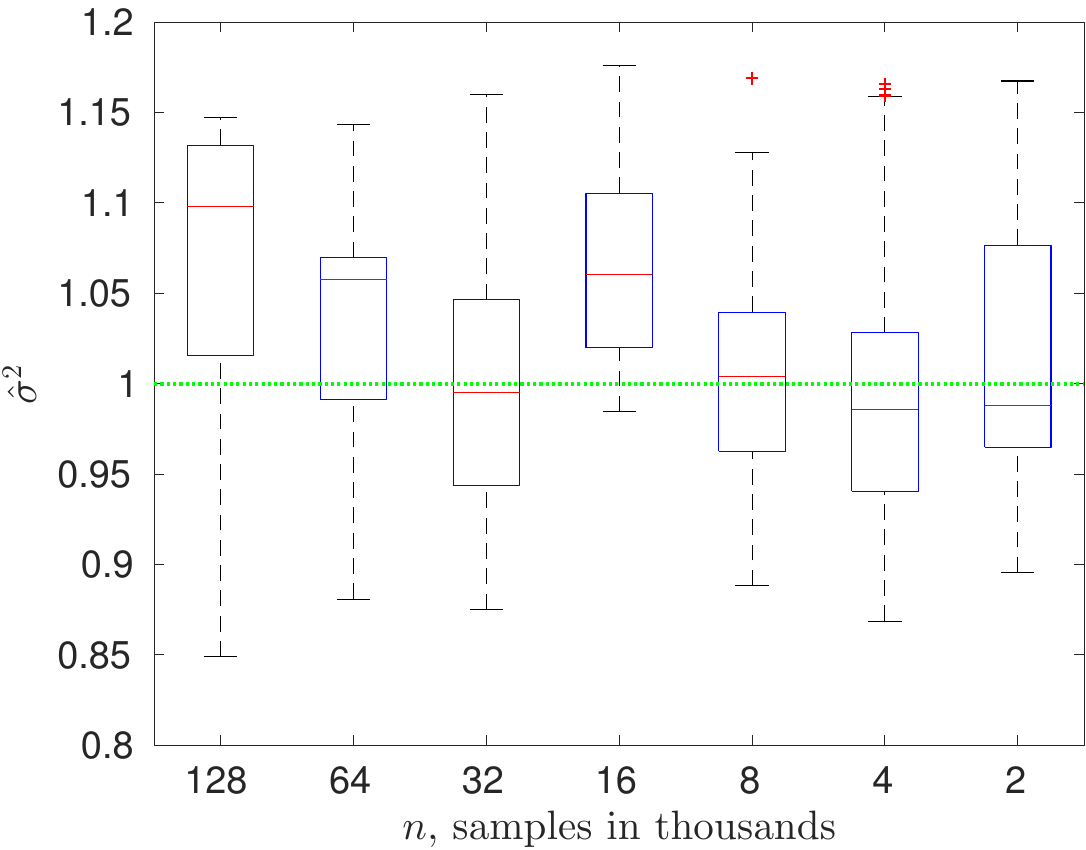}\vspace{-.1cm}
        \label{fig:sigma_synt}
    \end{subfigure}
    \begin{subfigure}[b]{0.48\textwidth}
     \centering
        \caption{}
       \includegraphics[width=8cm]{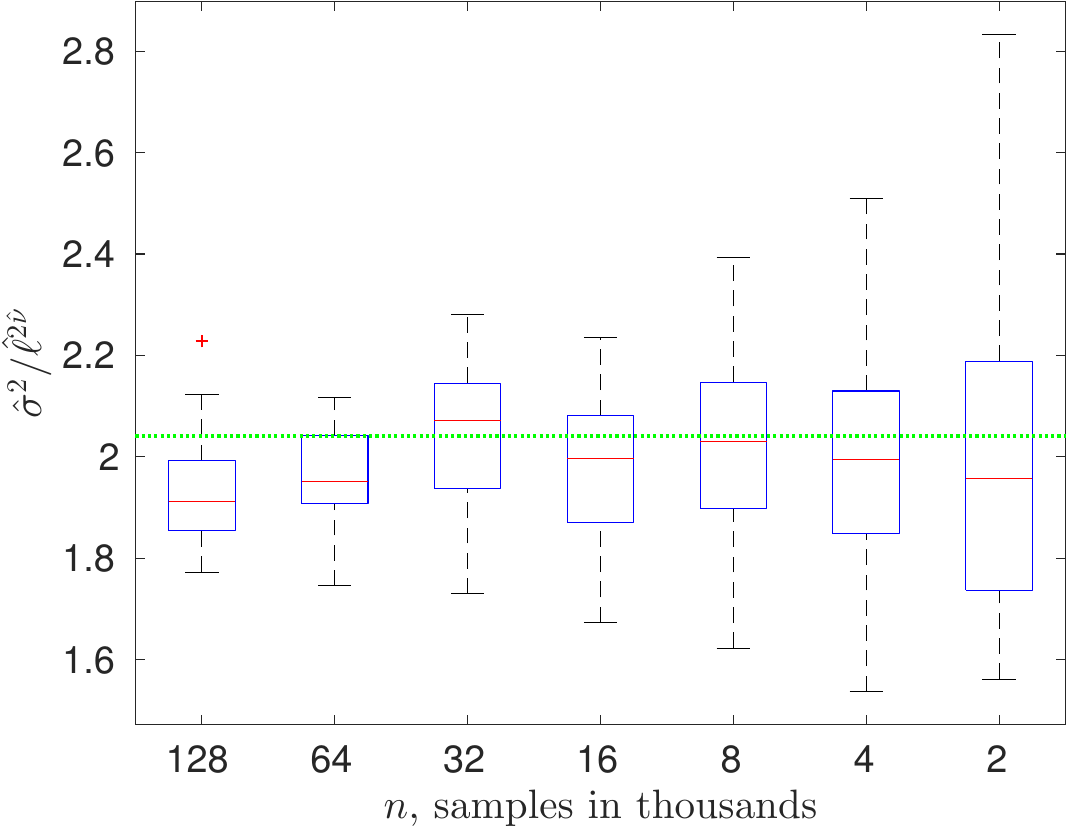}\vspace{-.1cm}
        \label{fig:coeff_synt}
    \end{subfigure}
\caption{  
Boxplots for (a) $\ell$; (b) $\nu$; (c) $\sigma^2$; and (d) $\sigma^2/\ell^{2\nu}$ for $n \in \{128, 64, 32,16, 8, 4, 2 \}\times  1{,}000$. Simulated data with known parameters $(\ell^*, \nu^*,\sigma^*)=(0.7, 0.9, 1.0)$. Boxplots are obtained from 50 replicates; the $\H$-matrix sub-block accuracy is $\varepsilon=10^{-7}$. The green, dotted horizontal line represents the true value of the parameters.}
\label{fig:synthetic_boxes}
\end{figure}
\textcolor{black}{This simulation study (Fig.~\ref{fig:synthetic_boxes}) shows boxplots vs $n$. We are able to estimate the unknown parameters of a Mat\'ern covariance function on an irregular grid with a certain accuracy. The parameter $\nu$ (Fig.~\ref{fig:nu_synt}) was identified more accurately than the parameters $\ell$ (Fig.~\ref{fig:ell_synt}) and $\sigma^2$ (Fig.~\ref{fig:sigma_synt}). 
It is difficult to say why we do not see a clear patterns for $\ell$ and $\sigma^2$.
We believe there are a few reasons. First, the iterative optimization method had difficulties to converge (we need to significantly decrease the threshold and increase the maximal number of iterations) for large $n$. This is very expensive regarding the computing time. Second, the log-likelihood is often very flat, and there are multiple solutions which fulfill the threshold requirements. We think that the reason is not in the $\H$-matrix accuracy since we took a very fine accuracy $\varepsilon=10^{-10}$ already. Third, we have fixed the locations for a given size $n$. Finally, under infill asymptotics, the parameters $\ell$, $\nu$, $\sigma^2$ cannot be estimated consistently but the quantity $\sigma^2/\ell^{2\nu}$ can, as seen in Fig.~\ref{fig:coeff_synt}.}

\textcolor{black}{
In Fig.~\ref{fig:50eps_fbplots} we present functional boxplots \cite{SunGenton11} of the estimated parameters as a function of the accuracy $\varepsilon$, based on 50 replicates with $n=\{4000, 8000, 32000\}$ observations. The true values to be identified are $\btheta^*=({\ell}, {\nu}, {\sigma^2},{\sigma^2/\ell^{2\nu}})=(0.7,0.9,1.0, 1.9)$, denoted by the doted green line. Functional boxplots on (a), (b), (c) identify ${\ell}$; on (d), (e), (f)  identify ${\nu}$; on (g), (h), (i) identify ${\sigma}^2$, and on (j), (k), (l) identify ${\sigma^2/\ell^{2\nu}}$.
One can see that we are able to identify all parameters with a good accuracy: the median, denoted by the solid black curve is very close to the dotted green line.
The parameter $\nu$ and the auxiliary coefficient ${\sigma^2/\ell^{2\nu}}$ are identified better than $\ell$ and $\sigma^2$, similarly to Fig.~\ref{fig:synthetic_boxes}. The size of the middle box in the boxplot indicates the variability of the estimates. It decreases when $n$ increases from $4{,}000$ to $32{,}000$.}

\textcolor{black}{The functional boxplots in Fig.~\ref{fig:50eps_fbplots} demonstrate that the estimates are fairly insensitive to the choice of the accuracy $\varepsilon$. Plots of the estimated parameters' curves for 30 replicated with $n=64{,}000$ are presented in Fig.~\ref{fig:50eps} in Appendix~\ref{app:B}.
All these plots suggest when to stop decreasing $\varepsilon$ further. The MLE estimates are already good enough for relatively large $\varepsilon$. In other words, the $\H$-matrix approximations are accurate enough, and to improve the estimates of the parameters, one should improve the optimization procedure (the initial guess, the threshold, the maximal number of iterations).}

%
\begin{figure}[htbp!]
\centering
    \centering
    \begin{subfigure}[b]{0.32\textwidth}
     \centering
        \caption{}
       \includegraphics[width=0.99\textwidth]{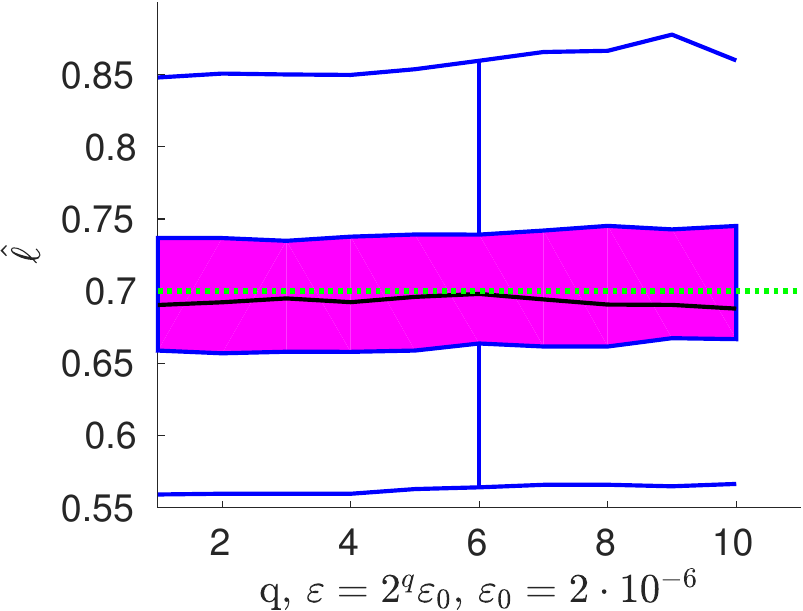}\vspace{-.1cm}
        \label{fig:ell4}
    \end{subfigure}
        \begin{subfigure}[b]{0.32\textwidth}
     \centering
        \caption{}
       \includegraphics[width=0.99\textwidth]{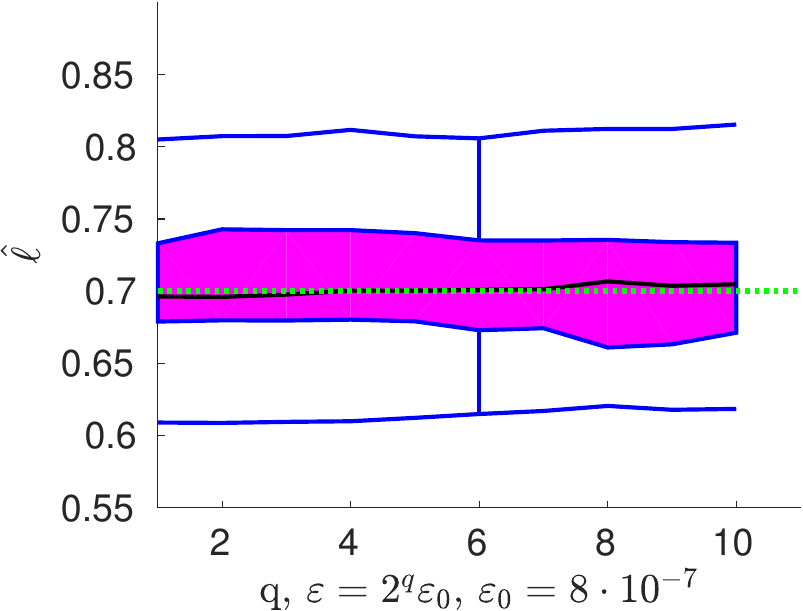}\vspace{-.1cm}
        \label{fig:ell8}
    \end{subfigure}
    \begin{subfigure}[b]{0.32\textwidth}
     \centering
        \caption{}
       \includegraphics[width=0.99\textwidth]{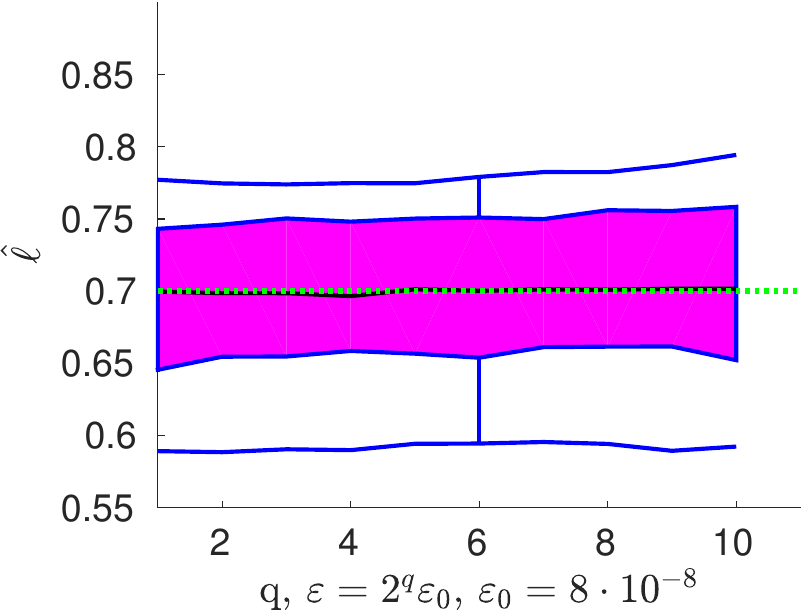}\vspace{-.1cm}
        \label{fig:ell64}
    \end{subfigure}
   \begin{subfigure}[b]{0.32\textwidth}
     \centering
        \caption{}
       \includegraphics[width=0.99\textwidth]{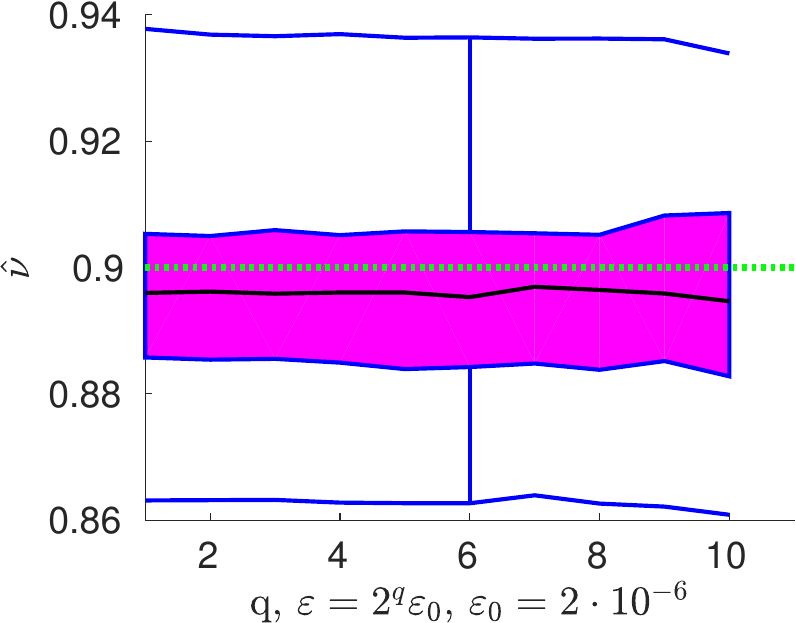}\vspace{-.1cm}
        \label{fig:nu4}
    \end{subfigure}
   \begin{subfigure}[b]{0.32\textwidth}
     \centering
        \caption{}
       \includegraphics[width=0.99\textwidth]{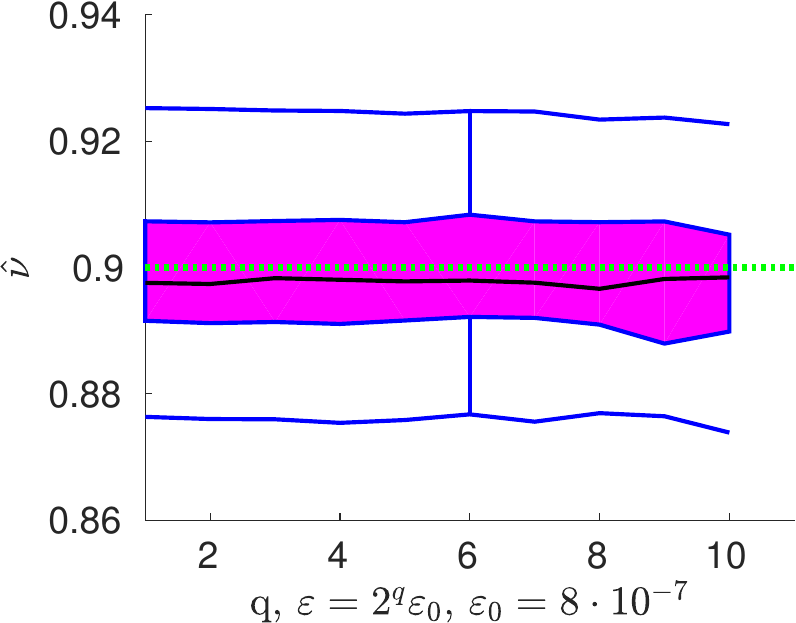}\vspace{-.1cm}
        \label{fig:nu8}
    \end{subfigure}
       \begin{subfigure}[b]{0.32\textwidth}
     \centering
        \caption{}
       \includegraphics[width=0.99\textwidth]{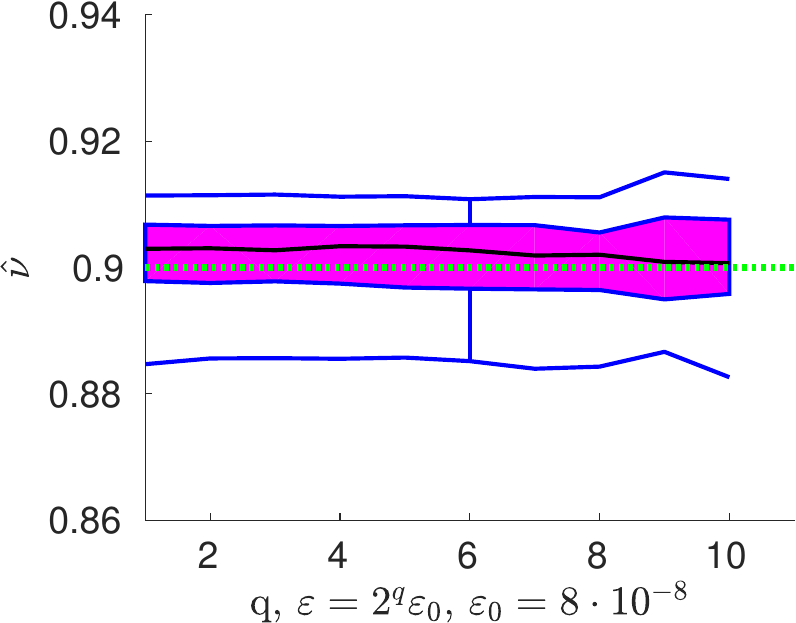}\vspace{-.1cm}
        \label{fig:nu64}
    \end{subfigure}
   \begin{subfigure}[b]{0.32\textwidth}
     \centering
        \caption{}
      \includegraphics[width=0.99\textwidth]{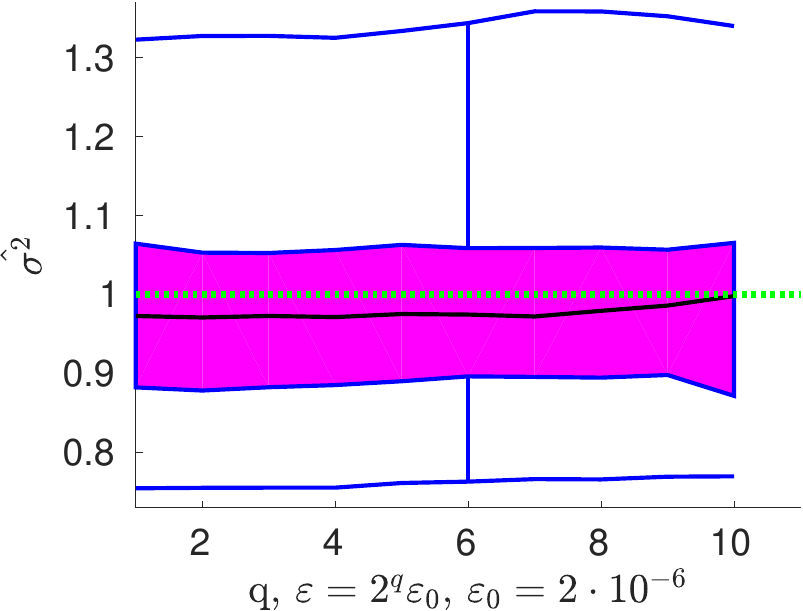}\vspace{-.1cm}
        \label{fig:sigma4}
    \end{subfigure}
   \begin{subfigure}[b]{0.32\textwidth}
     \centering
        \caption{}
      \includegraphics[width=0.99\textwidth]{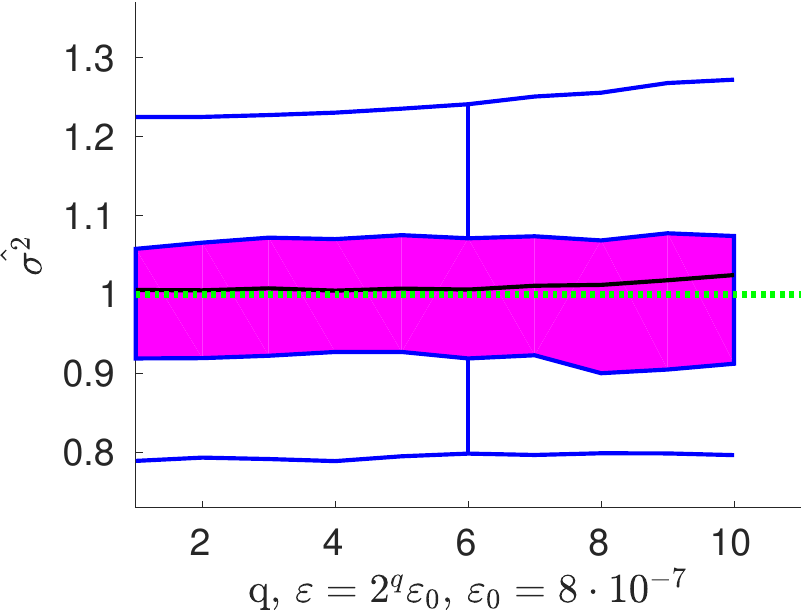}\vspace{-.1cm}
        \label{fig:sigma8}
    \end{subfigure}
   \begin{subfigure}[b]{0.32\textwidth}
     \centering
         \caption{}
     \includegraphics[width=0.99\textwidth]{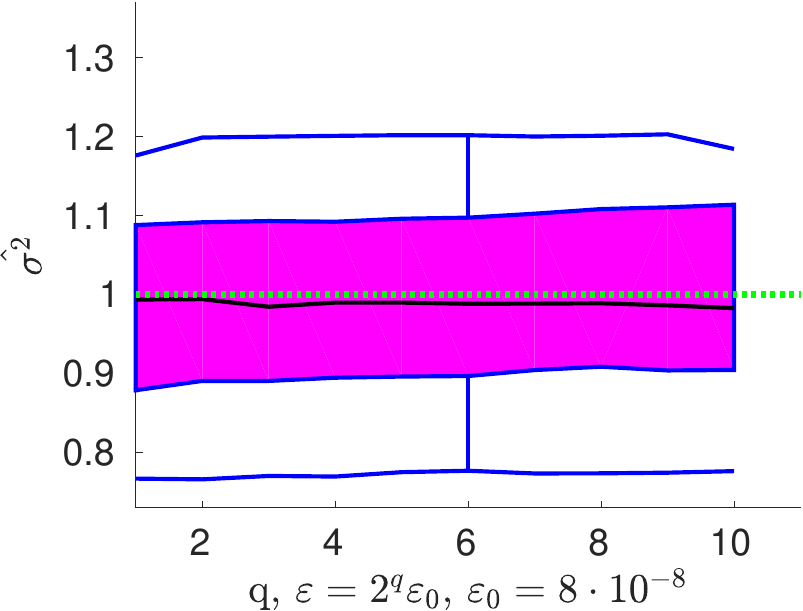}\vspace{-.1cm}
        \label{fig:sigma64}
    \end{subfigure}
   \begin{subfigure}[b]{0.32\textwidth}
     \centering
        \caption{}
      \includegraphics[width=0.99\textwidth]{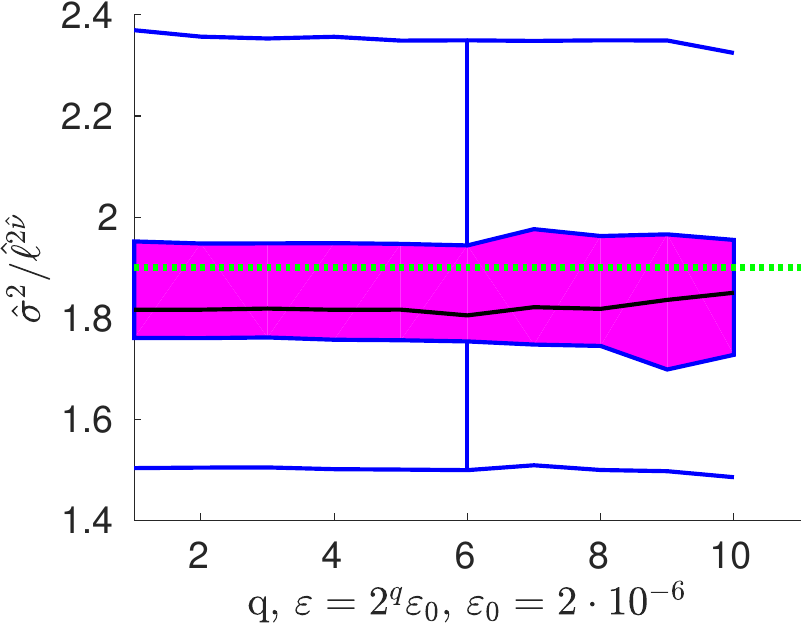}\vspace{-.1cm}
        \label{fig:sigma4}
    \end{subfigure}
   \begin{subfigure}[b]{0.32\textwidth}
     \centering
        \caption{}
      \includegraphics[width=0.99\textwidth]{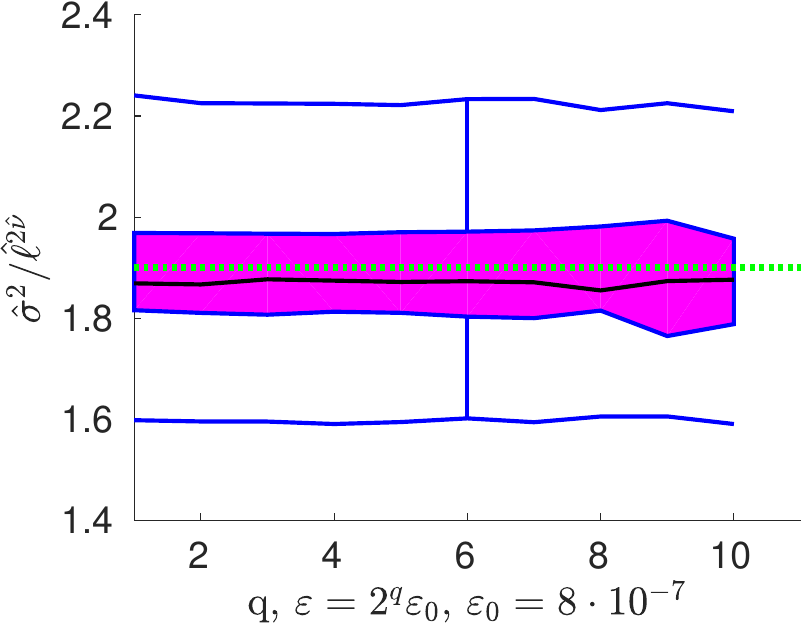}\vspace{-.1cm}
        \label{fig:sigma8}
    \end{subfigure}
   \begin{subfigure}[b]{0.32\textwidth}
     \centering 
         \caption{}
     \includegraphics[width=0.99\textwidth]{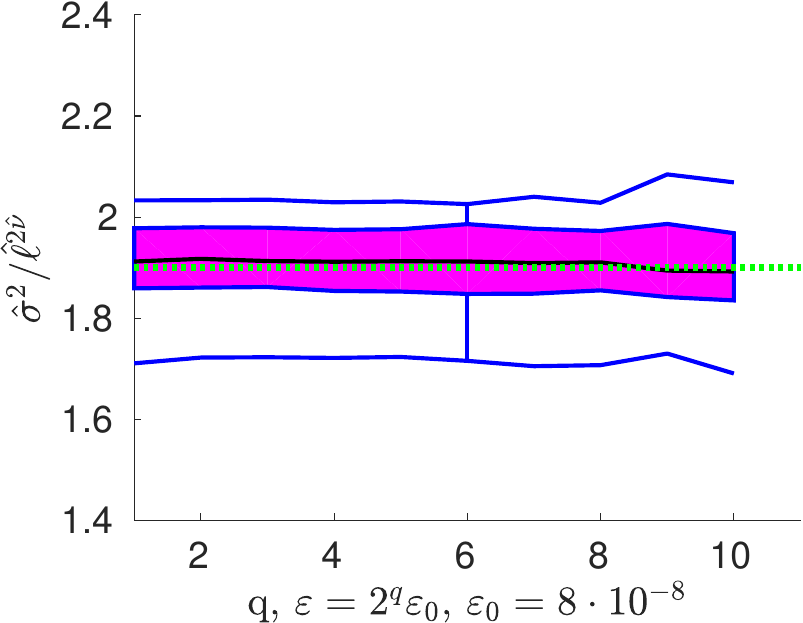}\vspace{-.1cm}
        \label{fig:sigma64}
    \end{subfigure}
\caption{Functional boxplots of the estimated parameters as a function of the accuracy $\varepsilon$, based on 50 replicates with $n=\{4000, 8000, 32000\}$ observations (left to right columns). True parameters $\btheta^*=({\ell}, {\nu}, {\sigma^2},{\sigma^2/\ell^{2\nu}})=(0.7,0.9,1.0, 1.9)$ represented by the green dotted lines. Replicates on (a), (b), (c) identify ${\ell}$; on (d), (e), (f)  identify ${\nu}$; on (g), (h), (i) identify ${\sigma}^2$, and on (j), (k), (l) identify ${\sigma^2/\ell^{2\nu}}$.}
\label{fig:50eps_fbplots}
\end{figure}
%
%
%
\section{Application to Soil Moisture Data}
\label{sec:moisture}
First, we introduce the daily soil moisture data set from a numerical model. Then
we demonstrate how to apply $\H$-matrices to fit a Mat\'ern covariance model. We use just one replicate sampled at $n$ locations. Then we perform numerical tests to see which $n$ is sufficient and explain how to chose the appropriate $\H$-matrix accuracy. We provide the corresponding computational times and storage costs.
\subsection{Problem and data description}
In climate and weather studies, numerical models play an important role in
improving our knowledge of the characteristics of the climate system and the causes behind climate variations. These numerical models describe the evolution of many variables, such as temperature, wind speed, precipitation, humidity, and pressure, by solving a set of equations. The process is usually very complicated, involving physical parameterization, initial condition configurations, numerical integration, and data output schemes. In this section, we use the proposed $\H$-matrix methodology to investigate the spatial variability of soil moisture data, as generated by numerical models.
Soil moisture is a key factor in evaluating the state of the hydrological process and has a broad range of applications in weather forecasting, crop yield prediction, and early warnings of flood and drought. It has been shown that a better characterization of soil moisture can significantly improve weather forecasting. However, numerical models often generate very large datasets due to the high spatial resolutions of the measured data, which makes the computation of the widely used Gaussian process models infeasible. Consequently, the whole region of interest must be divided into blocks of smaller size in order to fit the Gaussian process models independently to each block; alternatively, the size of the dataset must be reduced by averaging to a lower spatial resolution.  Compared to fitting a consistent Gaussian process model to the entire region, it is unclear how much statistical efficiency is lost by such an approximation. Since our proposed $\H$-matrix technique can handle large covariance matrix computations, and the parallel implementation of the algorithm significantly reduces the computational time, we are able to fit Gaussian process models to a set of large subsamples in a given spatial region. Therefore, next we explore the effect of sample size on the statistical efficiency.

We consider high-resolution soil moisture data from January 1, 2014, measured in the topsoil layer of the Mississippi River basin, U.S.A. The spatial resolution is 0.0083 degrees, and the distance of one-degree difference in this region is approximately 87.5 km. The grid consists of $1830 \times 1329 = 2{,} 432{,} 070$ locations with $2{,}000{,}000$ observations and $432{,} 070$ missing values. 
Therefore, the available spatial data are not on a regular grid. We use the same model for the mean process as in Huang and Sun (2018), and fit a zero-mean Gaussian process model with a Mat\'ern covariance function to the residuals; see Huang and Sun (2018) for more details on the data description and exploratory data analysis. 

\subsection{Estimation, computing times and required storage costs}
\label{sec:TimeMoist}
In the next experiment, we estimated the unknown parameters ${\ell}$, ${\nu}$, and ${\sigma^2}$ for different sample sizes $n$. Each time $n$ samples were chosen from the whole set randomly.
Table~\ref{table:approx_compare15} shows the computational time and storage for the $\H$-matrix approximations for $n=\{2000,\ldots, 2000000\}$. All computations are done with the parallel $\H$-matrix toolbox, HLIBpro. The number of computing cores is 40, the RAM memory 128GB.  
It is important to note that the computing time (columns 2 and 5) and the storage cost (columns 3 and 6) are growing nearly linearly with $n$. These numerical computations illustrate the theoretical formulas from Table~\ref{tab:1}. Additionally, we provide the accuracy of the $\H$-Cholesky inverse and estimated parameters. 
The choice of the starting value is important.
First, we run the optimization procedure with $n=2{,}000$ randomly sampled locations, and with the starting value $(\ell_0, \nu_0,\sigma_0^2)=(3, 0.5, 1)$, $\varepsilon=10^{-7}$, the  residual threshold $10^{-7}$ and the maximal number of iterations $1{,}000$. After $78$ iterations the threshold was achieved and the solution is $(\hat{\ell}, \hat{\nu}, \hat{\sigma}^2)=(2.71, 0.257, 1.07)$. This is the starting value for the experiment with $n=4{,}000$. The estimated parameters, calculated with $n=4{,}000$ and 80 iterations, are $(\hat{\ell}, \hat{\nu}, \hat{\sigma}^2)=(3, 0.228, 1.02)$. Second, we take these values as the new starting values for another randomly chosen $n=8{,}000$ locations. After 60 iterations, the residual threshold is again achieved, the new estimated values are $(\hat{\ell}, \hat{\nu}, \hat{\sigma}^2)=(3.6, 0.223, 1.03)$. We repeat this procedure for each new $n$ till $2{,}000{,}000$. At the end, for $n=2{,}000{,}000$ we obtain $(\hat{\ell}, \hat{\nu}, \hat{\sigma}^2)=(1.38, 0.34, 1.1)$.
\begin{table}[htbp!]
\centering
\caption{Estimation of unknown parameters for various $n$ (see columns 1,8,9,10). Computing times and storage costs are given for 1 iteration. The $\H$-matrix accuracy $\varepsilon=10^{-7}$.}
\begin{small}
\begin{tabular}{|c|ccc|ccc|ccc|}
\hline
$ n$ & \multicolumn{3}{c|}{$\widetilde{\bC}$} &  \multicolumn{3}{c|}{$\widetilde{\bL}\widetilde{ \bL}^\top$}&  \multicolumn{3}{c|}{\textbf{Parameters}} \\
             & time & size & kB/$n$  & time & size &  $\Vert \bI-(\widetilde{\bL}\widetilde{\bL}^\top)^{-1}{\bC} \Vert_2$  & $\hat{\ell} $&$\hat{\nu}$& $\hat{\sigma}^2$    \\ 
        &  sec. & MB &     & sec. &  MB &  & & &    \\ 
\hline
 2{,}000      &  0.04 & 5.2    &2.7 & 0.07   & 5.2    &$6.7\cdot 10^{-4}$   & 2.71 & 0.257 & 1.07 \\ 
 4{,}000      & 0.07  & 12.7  & 3.2& 0.17   & 13 & $1.5\cdot 10^{-4}$   & 3.03 & 0.228 & 1.02  \\  
 8{,}000      & 0.12  & 29.8  & 3.8& 0.36   & 32 & $3.1\cdot 10^{-4}$   & 3.62 & 0.223 & 1.03  \\  
 16{,}000    & 0.52  & 72.0  & 4.6& 1.11   & 75 &   $2.6\cdot 10^{-4}$   & 1.73 & 0.252 & 0.95  \\  
 32{,}000      & 0.93  & 193    & 6.2& 3.04   & 210 & $7.4\cdot 10^{-5}$   & 1.93 & 0.257 & 1.03  \\  
64{,}000         & 3.3 & 447   & 7.1 & 11.0 & 486 & $1.0\cdot 10^{-4}$ &1.48&0.272&1.01\\ 
128{,}000      & 7.7 & 1160   & 9.5 & 36.7 & 1310 & $3.8\cdot 10^{-5}$  &0.84& 0.302& 0.94\\ 
256{,}000      & 13 & 2550   & 10.5 & 64.0 & 2960 & $7.1\cdot 10^{-5}$   &1.41&0.327& 1.24\\ 
512{,}000      & 23 & 4740   & 9.7 & 128 & 5800 & $7.1\cdot 10^{-4}$  &1.41& 0.331 & 1.25\\ 
1{,}000{,}000 & 53 & 11260   & 11 & 361 & 13910 & $3.0\cdot 10^{-4}$   &1.29&0.308&1.00\\ 
2{,}000{,}000 & 124 & 23650   & 12.4 & 1001 & 29610 & $5.2\cdot 10^{-4}$   &1.38& 0.340& 1.10\\ 
\hline
\end{tabular}
\end{small}
\label{table:approx_compare15}  
\end{table}

In Table~\ref{table:costs_vs_eps} we list the computing time and the storage cost (total and divided by $n$) vs. $\H$-matrix accuracy $\varepsilon$ in each sub-block. This table connects the $\H$-matrix sub-block accuracy, the required computing resources, and the   
inversion error $\Vert \bI-(\widetilde{\bL}\widetilde{\bL}^\top)^{-1}{\bC} \Vert_2$. It provides some intuition about how much computational resources are required for each $\varepsilon$ for fixed $n$.
We started with $\varepsilon=10^{-7}$ (the first row) and then multiplied it each time by a factor of $4$.
The last row corresponds to $\varepsilon=6.4\cdot 10^{-3}$, which is sufficient to compute $\widetilde{\bC}$, but is not sufficient to compute the $\H$-Cholesky factor $\widetilde{\bL}$ (the procedure for computing $\H$-Cholesky crashes). We see that the storage and time increase only moderate with decreasing $\varepsilon$.

\begin{table}[h!]
\centering
\caption{Computing time and storage cost vs. accuracy $\varepsilon$ for parallel $\H$-matrix approximation; number of cores is 40, $\hat{\nu}=0.302$, $\hat{\ell}=0.84$, $\hat{\sigma}^2=0.94$, $n=128{,}000$. $\H$-matrix accuracy in each sub-block for both $\widetilde{\bC}$ and $\widetilde{\bL}$ is $\varepsilon$.}
\begin{small}
\begin{tabular}{|c|ccc|ccc|}
\hline
$ \varepsilon$ & \multicolumn{3}{c|}{$\widetilde{\bC}$} &  \multicolumn{3}{c|}{$\widetilde{\bL}\widetilde{ \bL}^\top$}  \\
             & time & size, MB & kB/$n$  & time & size &  $\Vert \bI-(\widetilde{\bL}\widetilde{\bL}^\top)^{-1}\bC \Vert_2$     \\ 
        &  sec. & GB &     & sec. & kB/dof  &      \\ 
\hline
$1.0\cdot 10^{-7}$   & 5.8 & 1110     & 9.16 &  21.5 & 10.0 & $8.4\cdot 10^{-5}$  \\ 
$4.0\cdot 10^{-7}$   & 5.4 & 1021      & 8.16 &  16.2 & 9.0 & $3.5\cdot 10^{-4}$  \\ 
$1.6\cdot 10^{-6}$   & 4.5 & 913      & 7.3 &  13.3 & 8.1 & $1.3\cdot 10^{-3}$  \\ 
$6.4\cdot 10^{-6}$   & 4.0 & 800      & 6.4 &  10.7 & 7.3 & $4.4\cdot 10^{-3}$  \\ 
$2.6\cdot 10^{-5}$   & 3.5 & 698      & 5.6 &   8.7  & 6.3 & $2.0\cdot 10^{-2}$  \\ 
$1.0\cdot 10^{-4}$   & 3.1 & 625      & 5.0 &   7.0  & 5.5 & $6.4\cdot 10^{-2}$  \\ 
$4.1\cdot 10^{-4}$   & 2.7 & 550      & 4.4 &   6.2  & 5.1 & $2.8\cdot 10^{-1}$  \\ 
$1.6\cdot 10^{-3}$   & 2.4 & 467      & 3.7 &   5.0  & 4.5 & $2.5$  \\ 
$6.4\cdot 10^{-3}$   & 2.3 & 403      & 3.2 &   -  & - & - \\ 
\hline
\end{tabular}
\end{small}
\label{table:costs_vs_eps}  
\end{table}

In Table~\ref{table:costs_2MI} we use the whole daily moisture data set with $2\cdot 10^{6}$ locations to estimate the unknown parameters. We provide computing times to set up the matrix $\widetilde{\bC}$ and to compute its $\H$-Cholesky factorization $\widetilde{\bL}$. Additionally, we list the total storage requirement and the storage divided by $n$ for both $\widetilde{\bC}$ and $\widetilde{\bL}$. 
The initial point for the optimization algorithm is $(\ell_0, \nu_0, \sigma_0^2)=(1.18, 0.42, 1.3)$. \textcolor{black}{We note that the choice of the initial point is very important and may have a strong influence on the final solution (due to the fact that the log-likelihood could be very flat around the optimum).  We started our computations with a rough $\H$-matrix accuracy $\varepsilon=10^{-4}$ and the $\H$-Cholesky procedure crashes. We tried $\varepsilon=10^{-5}$ and received  $(\hat{\ell}, \hat{\nu}, \hat{\sigma})=(1.39, 0.325, 1.01)$. We took them as the initial condition for $\varepsilon=10^{-5}$ and received very similar values, and so on.
To find out which $\H$-matrix accuracy is sufficient, we provide tests for different values of  $\varepsilon \in \{10^{-4}, 10^{-5}, 10^{-6}, 10^{-7}, 10^{-8}\}$. One can see that $\varepsilon=10^{-4}$ is not sufficient (the $\H$-Cholesky procedure crashes), and $\varepsilon=10^{-5}$ provides similar results as $\varepsilon \in \{10^{-6}, 10^{-7}, 10^{-8}\}$. Therefore, our recommendation would be to start with some rather low $\H$-matrix accuracy, to use the obtained $\hat{\btheta}$ as the new initial guess, and then decrease $\varepsilon$ and repeat this a few times.}
\begin{table}[h!]
\centering
\caption{Estimating unknown parameters for various $\H$-matrix accuracies $\varepsilon$; number of cores is 40, $n=2{,}000{,}000$. $\H$-matrix accuracy in each sub-block for both $\widetilde{\bC}$ and $\widetilde{\bL}$ is $\varepsilon$, max. number of iterations is 60, threshold in the iterative solver is $10^{-4}$.}
\begin{small}
\begin{tabular}{|c|ccc|ccc|cccc|}
\hline
$ \varepsilon$ &  \multicolumn{3}{c|}{Estimated parameters} &  \multicolumn{3}{c|}{Costs, $\widetilde{\bC}$}&  \multicolumn{4}{c|}{Costs, $\widetilde{\bL}$}  \\ \hline
             & $\hat{\ell}$ & $\hat{\nu}$ & $\hat{\sigma}^2$  & time, &size, &size/$n$ & time, &size,&size/$n$, &  $\Vert \bI-(\widetilde{\bL}\widetilde{\bL}^\top)^{-1}\bC \Vert_2$     \\ 
            & & &   & set up & GB & kB & sec.& GB &kB &      \\ \hline
$10^{-4}$   & -& -      & - & 58  & 11 & 6 &  - & -& -   & -\\ 
$10^{-5}$  &1.39 & 0.325        & 1.01 & 61  & 12 & 6 &  241   & 15.4& 8.11   & $8.1\cdot 10^{-2}$\\ 
$10^{-6}$   & 1.41& 0.323      & 1.04& 66  & 15 & 7.7 & 379  & 20.1&  10.3  & $5.7\cdot 10^{-2}$\\ 
$10^{-7}$   & 1.39& 0.330      & 1.05& 114  & 22 &  11&  601 & 27.4&  14.0 & $2.7\cdot 10^{-3}$\\ 
$10^{-8}$   & 1.38& 0.329      & 1.06& 138  & 29 & 15 & 1057  & 34.4&  18.0  & $1.5\cdot 10^{-4}$\\ 
\hline
\end{tabular}
\end{small}
\label{table:costs_2MI}  
\end{table}
\subsection{Reproducibility of the numerical results}
To reproduce the presented numerical simulations, the first step is downloading the HLIB (www.hlib.org) or HLIBPro (www.hlibpro.com) libraries. Both are free for academic purposes. HLIB is a sequential, open-source C-code library, which is relatively easy to use and understand. HLIBPro is the highly tuned parallel library, where only the header and object files are accessible.

After signing the $\H$-matrix (HLIB or HLIBPro) license, downloading, and installing the libraries, our modules for approximating covariance matrices and computing log-likelihoods, and the soil moisture data can be downloaded from $\mbox{https://github.com/litvinen/HLIBCov.git}$.

HLIB requires not only the set of locations, but also the corresponding triangulation of the computing domains, \cite{MYPHD, khoromskij2008domain}. The vertices of the triangles should be used as the location points. To construct a triangulation, we recommend employing the routines in MATLAB, R, or any other third-party software. HLIBPro does not require triangulation; only the coordinates of the locations are needed \cite{litvHLIBPro17}.
%
%
%
\section{Conclusion and Discussion}
\label{sec:Conclusion}
We have applied a well-known tool from linear algebra, the $\H$-matrix technique, to 
spatial statistics. This technique makes it possible to work with large datasets of measurements observed on unstructured grids, in particular, to estimate the unknown parameters of a covariance model. The statistical model considered yields Mat\'{e}rn covariance matrices with three unknown parameters,  $\ell$, $\nu$, and $\sigma^2$. 
We applied the $\H$-matrix technique in approximating multivariate Gaussian log-likelihood functions and Mat\'{e}rn covariance matrices. 
The $\H$-matrix technique allowed us to drastically reduce the required memory 
and computing time. 
Within the $\H$-matrix approximation,
we can increase the spatial resolution, take more measurements into account, and consider larger regions with larger data sets. $\H$-matrices require neither axes-parallel grids nor homogeneity of the covariance function.

From the $\H$-matrix approximation of the log-likelihood, we computed the $\H$-Cholesky factorization, the KLD, the log-determinant, and a quadratic form, $\bZ^{\top}\bC^{-1}\bZ$ (Tables~\ref{table:eps_det}, \ref{table:approx_compare_rank}). We demonstrated the computing time, storage requirements, relative errors, and convergence of the $\H$-matrix technique (Tables~\ref{table:approx_compare15}, \ref{table:costs_vs_eps}, \ref{table:costs_2MI}).

We reduced the cost of computing the log-likelihood from cubic to log-linear, 
by employing the $\H$-matrix approximations, without significantly changing the log-likelihood.
We considered both simulated examples, where we identified the known parameters, and a large, real data (soil moisture) example, where the parameters were unknown.
We were able to calculate the maximum likelihood estimate ($\hat{\ell}$, $\hat{\nu}$, $\hat{\sigma}^2$) for all examples. We researched the impact of the $\H$-matrix approximation error and the statistical error (see Fig.~\ref{fig:50eps}) to the total error. For both examples we repeated the calculations for 50 replicates and computed box plots. We analyzed the dependence of ($\hat{\ell}$, $\hat{\nu}$, $\hat{\sigma}^2$) on $\varepsilon$ and on the sample size $n$.

With the parallel $\H$-matrix library HLIBPro, we were able to compute the log-likelihood function for $2{,}000{,}000$ locations in a few minutes (Fig.~\ref{table:approx_compare15}) on a desktop machine, which is $\approx 5$ years old and cost nowadays $\approx 5.000$ USD. At the same time, computation of the maximum likelihood estimates is much more expensive and depends on the number of iterations in the optimization and can take from few hours to few days. In total, the algorithm may need 100-200 iterations (or more, depending on the initial guess and the threshold). If each iteration takes 10 minutes, then we may need 24 hours to get ($\hat{\ell}$, $\hat{\nu}$, $\hat{\sigma}^2$). Possible extension of this work are 1) to reduce the number of required iterations by implementing the first and second derivatives of the likelihood; 2) add nugget to the set of unknown parameters; 3) combine the current framework with the domain decomposition method.\\

\textbf{Acknowledgment}
The research reported in this publication was supported by funding from King Abdullah University of Science and Technology (KAUST). Additionally, we would like to express our special thanks of gratitude to Ronald Kriemann (for the HLIBPro software library) as well as to Lars Grasedyck and Steffen Boerm for the HLIB software library. Alexander Litvinenko was supported by the Bayesian Computational Statistics $\&$ Modeling group (KAUST), and by the SRI-UQ group (KAUST).\\

\textbf{References}
\bibliography{Hcovariance}

\begin{appendices}
\section{Appendix: Error estimates}
\label{appendix:A}


%
%
%
%
The Lemma 3.3. (page 5 in \cite{BallaniKressner}) gives the following result.
Let $\bC\in \mathbb{R}^{n\times n}$, and $\bE:=\bC-\widetilde{\bC}$, $\widetilde{\bC}^{-1}\bE:=\widetilde{\bC}^{-1} \bC- \bI$, and for the spectral radius\\
$\rho(\widetilde{\bC}^{-1} \bE)=\rho(\widetilde{\bC}^{-1}\bC-\bI) <\varepsilon<1$.
Then
$$\vert \log \mydet{\bC} - \log \mydet{\widetilde{\bC}}\vert \leq -n \log(1-\varepsilon).$$


To prove Theorem~\ref{thm:MainLL} we use the result above and
\begin{align*}
\vert \widetilde{\LL}(\btheta;k) - \LL(\btheta) \vert &= \frac{1}{2}\log\frac{\mydet \bC}{\mydet{ \widetilde{\bC}}}-  \frac{1}{2}\vert \bz^\top \left(\bC^{-1} - \widetilde{\bC}^{-1}\right )\bz \vert \\
&\leq -\frac{1}{2}\cdot n \log(1-\varepsilon)-  \frac{1}{2}\vert \bz^\top \left(\bC^{-1}\bC - \widetilde{\bC}^{-1}\bC\right )\bC^{-1}\bz \vert \\
&\leq -\frac{1}{2}\cdot n \log(1-\varepsilon)-  \frac{1}{2}\vert \bz^\top \left(\bI - \widetilde{\bC}^{-1}\bC\right )\bC^{-1}\bz \vert \\
&\leq  \frac{1}{2}\cdot n \varepsilon + \frac{1}{2}\Vert \bz \Vert^2\cdot c_1\cdot \varepsilon,
\end{align*}
where $-\log(1-\varepsilon)\approx \varepsilon$ for small $\varepsilon$.
\begin{rem}
\label{rem:NormZ}
Let $\bZ \sim \mathcal{N}(\mathbf{0},\sigma^2\bI)$.
Theorem 2.2.7 in \cite{Nickl_2015} provides estimates for a norm of a Gaussian vector, $\Vert \bZ\Vert$, for instance for $\sup_{t\in T}\vert Z(t) \vert$,
\begin{equation*}
Pr \left \{ \left \vert \sup_{t\in T}\vert Z(t)  \vert-M \right \vert >u \right \} \leq 2\left (1-\Phi \left (\frac{u}{\sigma} \right) \right)\leq e^{-\frac{u^2}{2\sigma^2}}, \quad M>0.%
\end{equation*}
\end{rem}
\begin{rem}
The assumption $\Vert \bC^{-1} \Vert \leq c_1$ is strong. This estimation depends on the matrix $\bC$, the smoothness parameter $\nu$, the covariance length $\ell$, and the $\H$-matrix rank $k$.  First, we find an appropriate block-decomposition for the covariance matrix $\bC$. Second, we estimate the ranks of $\widetilde{\bC}$. Third, we prove that the inverse/Cholesky can also be approximated in the $\H$-matrix format. Then we estimate the ranks for the inverse $\widetilde{\bC}^{-1}$ and the Cholesky factor $\widetilde{\bL}$. Finally, we estimate the $\H$-matrix approximation accuracies; see \cmmnt{\citep}\cite{HackHMEng}. 
In the worst case, the rank $k$ will be of order $n$. We also note that some covariance matrices are singular, so that $\widetilde{\bC}^{-1}$ and $\widetilde{\bL}$ may not exist. The computation of $\log\mydet{\widetilde{\bC}}$ could be an ill-posed problem, in the sense that small perturbations in $\bC$ result in large changes in
$\log\mydet{\bC}$. 
\end{rem}

Below we estimate the error caused by the usage of $\H$-matrices to generate simulated data.
Let $\widetilde{\bC}=\bC+\varepsilon \bC$ be an $\H$-matrix approximation, and $\bZ=\bL\bW$ be the data generated without $\H$-matrix approximation.
\begin{lemma} 
\label{lem:errZ}
Let  $\widetilde{\bZ}=\widetilde{\bL}\bW$, and $\Vert \widetilde{\bL}-\bL\Vert \leq \varepsilon_L$, where $\varepsilon_L$ tends to zero when the $\H$-matrix ranks are growing. Let $\bZ=\bz$, $\tilde{\bZ}=\tilde{\bz}$, $\bW=\bw$ denote the data realizations. Then  
\begin{equation*}
\Vert \widetilde{\bz}-\bz \Vert \leq \Vert \widetilde{\bL}-\bL\Vert \Vert \bw \Vert \leq \varepsilon_L \Vert \bw \Vert.
\end{equation*}
\end{lemma} 
The error $\varepsilon_L$ depends on the condition number of $\bC$ and can be large, for instance, for ill-conditioned matrices. 
The next lemma estimates the error in the quadratic form by replacing the original data set $\bZ=\bz $ by its $\H$-matrix approximation.
\begin{lemma}
\label{lem:L}
Let $\bL=\widetilde{\bL}+\varepsilon_L \bL$ be an $\H$-matrix approximation of $\bL$, then 
$$\vert \widetilde{\bz}^\top {\bC}^{-1} \widetilde{\bz} - \bz^\top \bC^{-1} \bz\vert \leq 2\varepsilon \vert  \widetilde{\bz}^\top {\bC}^{-1} \widetilde{\bz}  \vert +\mathcal{O}(\varepsilon^2).$$
\end{lemma}
\begin{proof}
Let $\widetilde{\bZ}=\widetilde{\bL}\bW$, $\bW\sim \mathcal{N}(\mathbf{0},\bI)$, and ${\bZ} = {\bL}\bW = (\widetilde{\bL} + \varepsilon_L\widetilde{\bL} )\bW=\widetilde{\bZ}+\varepsilon_L \widetilde{\bZ}$. Let $\bZ=\bz$, and $\tilde{\bZ}=\tilde{\bz}$ denote the data realizations.
Then, after  simple calculations, we obtain
\begin{equation*}
 \vert \widetilde{\bz}^\top {\bC}^{-1} \widetilde{\bz} - (\widetilde{\bz} +\varepsilon_L \widetilde{\bz})^\top \bC^{-1} (\widetilde{\bz} +\varepsilon_L \widetilde{\bz}) \vert \leq  2\varepsilon_L \vert  \widetilde{\bz}^\top {\bC}^{-1} \widetilde{\bz}  \vert + \varepsilon_L^2\vert  \widetilde{\bz}^\top {\bC}^{-1} \widetilde{\bz}  \vert \rightarrow 0,\; \text{if }\; \varepsilon_L \rightarrow 0.
\end{equation*}
\end{proof}
\begin{lemma}
\label{lem:inv}
Let $\bz\in \mathbb{R}^{n}$ and
$\Vert \widetilde{\bC}^{-1} - \bC^{-1} \Vert\leq \varepsilon_c$, then 
$\vert {\bz}^\top \widetilde{\bC}^{-1}{\bz} - \bz^\top \bC^{-1} \bz\vert \leq \Vert \bz\Vert^2 \varepsilon_c$.
\end{lemma}
\begin{proof}
\begin{equation*}
\vert {\bz}^\top \widetilde{\bC}^{-1}{\bz} - \bz^\top \bC^{-1} \bz\vert 
=
\vert {\bz}^\top( \widetilde{\bC}^{-1} - \bC^{-1} )\bz\vert \leq \Vert \bz\Vert^2 \varepsilon_c.
\end{equation*}
\end{proof}
\begin{lemma}
\label{lem:error1}
Let $\bz,\widetilde{\bz} \in \mathbb{R}^{n}$, 
and conditions of Lemmas~\ref{lem:L}-\ref{lem:inv} hold, then
$$\vert \widetilde{\bz}^\top \widetilde{\bC}^{-1} \widetilde{\bz} - \bz^\top \bC^{-1} \bz\vert \leq \mathcal{O}(\varepsilon_L) + \mathcal{O}(\varepsilon_c).$$
\end{lemma}
\begin{proof}
Adding and subtracting  an additional term and applying the previous Lemmas~\ref{lem:L}-\ref{lem:inv}, we obtain 
\begin{align*}
\vert \widetilde{\bz}^\top \widetilde{\bC}^{-1} \widetilde{\bz} - \bz^\top \bC^{-1} \bz\vert &\leq 
\vert \widetilde{\bz}^\top \widetilde{\bC}^{-1} \widetilde{\bz} - \bz^\top \widetilde{\bC}^{-1} \bz\vert + \vert \bz^\top \widetilde{\bC}^{-1} \bz- \bz^\top \bC^{-1} \bz\vert \leq 2\varepsilon_L \vert  \widetilde{\bz}^\top \widetilde{\bC}^{-1} \widetilde{\bz}  \vert  + \Vert \bz \Vert^2 \varepsilon_c .
\end{align*}
\end{proof}
Thus, if the approximations of the matrices $\bC$, $\bC^{-1}$ and $\bL$ are accurate (the error tends to zero relatively fast with increasing the ranks), then the $\H$-matrix approximation error in the quadratic form also tends to zero with increasing $k$ or decreasing $\varepsilon_L$. A more rigorous proof calculating the order of convergence is out of the scope of this paper. Dependence on the condition number is researched in \cite{BebendorfSpEq16, bebendorf2007Why} (only for matrices that origin from elliptic partial differential equations and integral equations).

In Table~\ref{tab:Tnorms} we list the Frobenius and spectral norms of $\widetilde{\bC}$, $\widetilde{\bL}$, and $\widetilde{\bC}^{-1}$ vs. $n$. This table shows that $\Vert \tilde{\bC} \Vert_2$ is growing with $n$ and $\Vert \tilde{\bC}^{-1} \Vert_2$ almost not growing with $n$.
\begin{table}[]
\begin{center}
\caption{Norms of Mat\'ern covariance matrix $\widetilde{\bC}$ and its Cholesky factor $\widetilde{\bL}$, $\ell=0.089$, $\nu=0.22$, 
$\sigma^2=1$, relative $\H$-matrix block-wise accuracy is $\varepsilon=10^{-5}$.}
\begin{tabular}{|c|c|c|c|c|c|c|c|}
\hline
 $n$ & $\Vert \widetilde{\bC} \Vert_F$& $\Vert \widetilde{\bC} \Vert_2$ & $\Vert \widetilde{\bL} \Vert_F$& $\Vert\widetilde{ \bL} \Vert_2$& $\Vert \widetilde{\bC}^{-1} \Vert_2$ \\
 \hline
 $32\cdot 10^3$ & $241$     & $7.70$     & $182$  & $3.2$ & $2.90$\\ \hline
 $64\cdot 10^3$ & $410$     & $13.5$   & $262$  & $4.5$  & $2.97$\\ \hline
 $128\cdot 10^3$ & $739$   & $24.0$      & $379$  & $6.4$ & $3.05$ \\ \hline
 $256\cdot 10^3$ & $1388$   & $46.0$          & $552$        & $9.2$&  $3.08$\\ \hline
\end{tabular}
\label{tab:Tnorms}
\end{center}
\end{table}
\section{Appendix: Admissibility condition and cluster and block-cluster trees}
\label{sec:adm}
The admissibility condition (criteria) is used to divide a given matrix into sub-blocks and to define which sub-blocks can be approximated well by low-rank matrices and which not.  There are many different admissibility criteria (see Fig.~\ref{fig:Hexample_adm}). The typical criteria are: the strong, weak and domain decomposition-based  \cite{HackHMEng}. The user can also develop his own admissibility criteria, which may depend, for example, on the covariance length. It regulates, for instance, block partitioning, the size of the largest block, the depth of the hierarchical block partitioning. Large low-rank blocks are good, but they may require also larger ranks.
Figure~\ref{fig:Hexample_adm} shows three examples of $\H$-matrices with three different admissibility criteria: (left) standard, (middle) domain-decomposition based and (right) weak (HODLR).  Matrices are taken from different applications and illustrate only the diversity of block partitioning.

\begin{figure}[htbp!]
\centering
\includegraphics[width=0.28\textwidth]{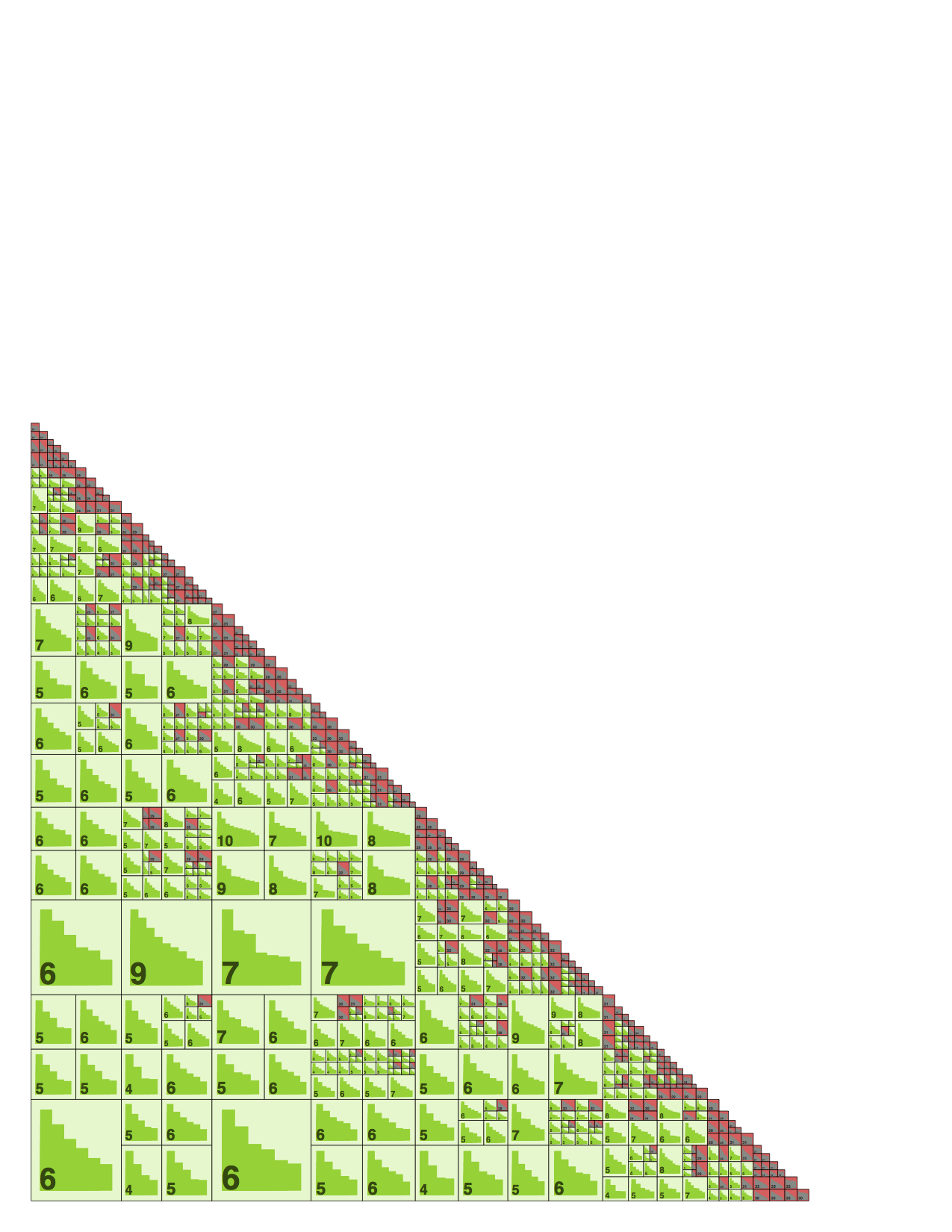}
\includegraphics[width=0.28\textwidth]{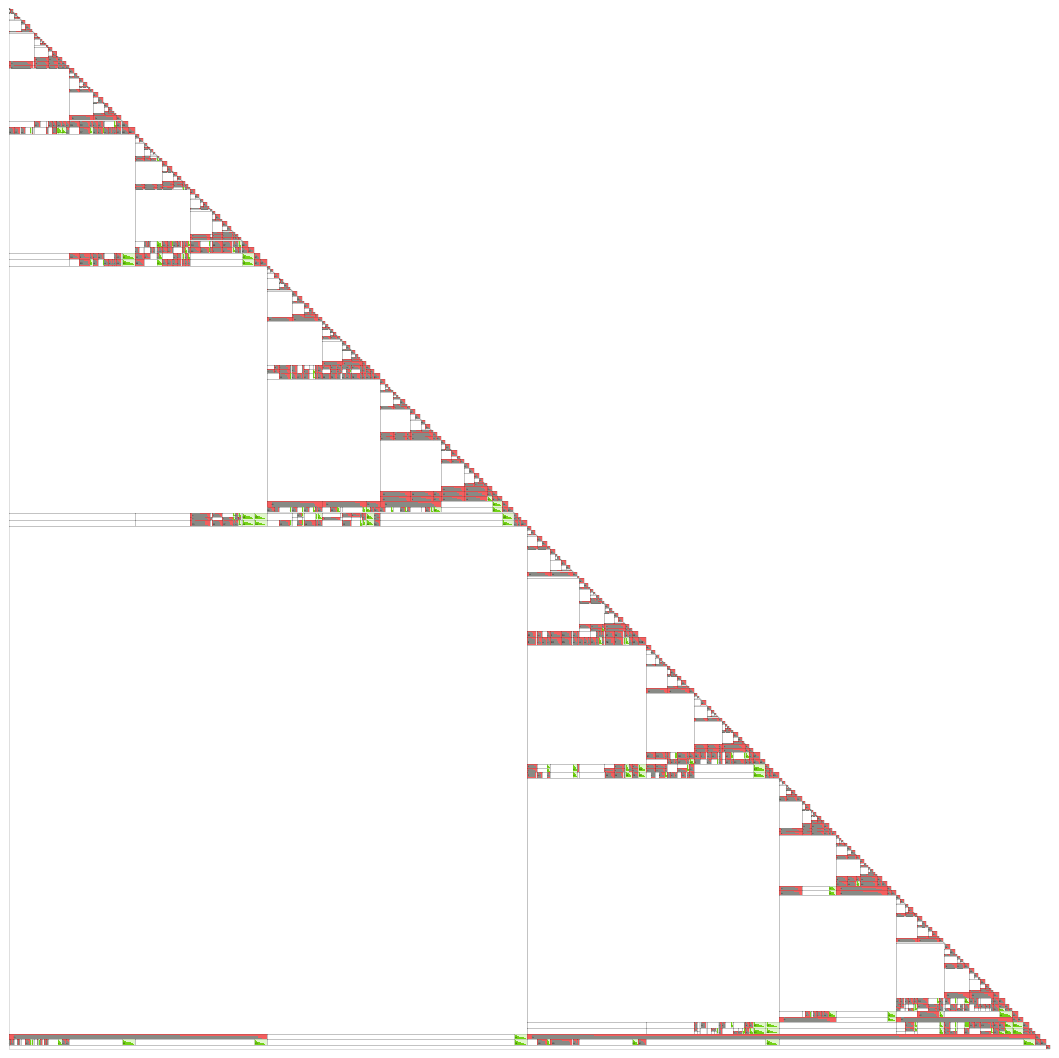}
\includegraphics[width=0.28\textwidth]{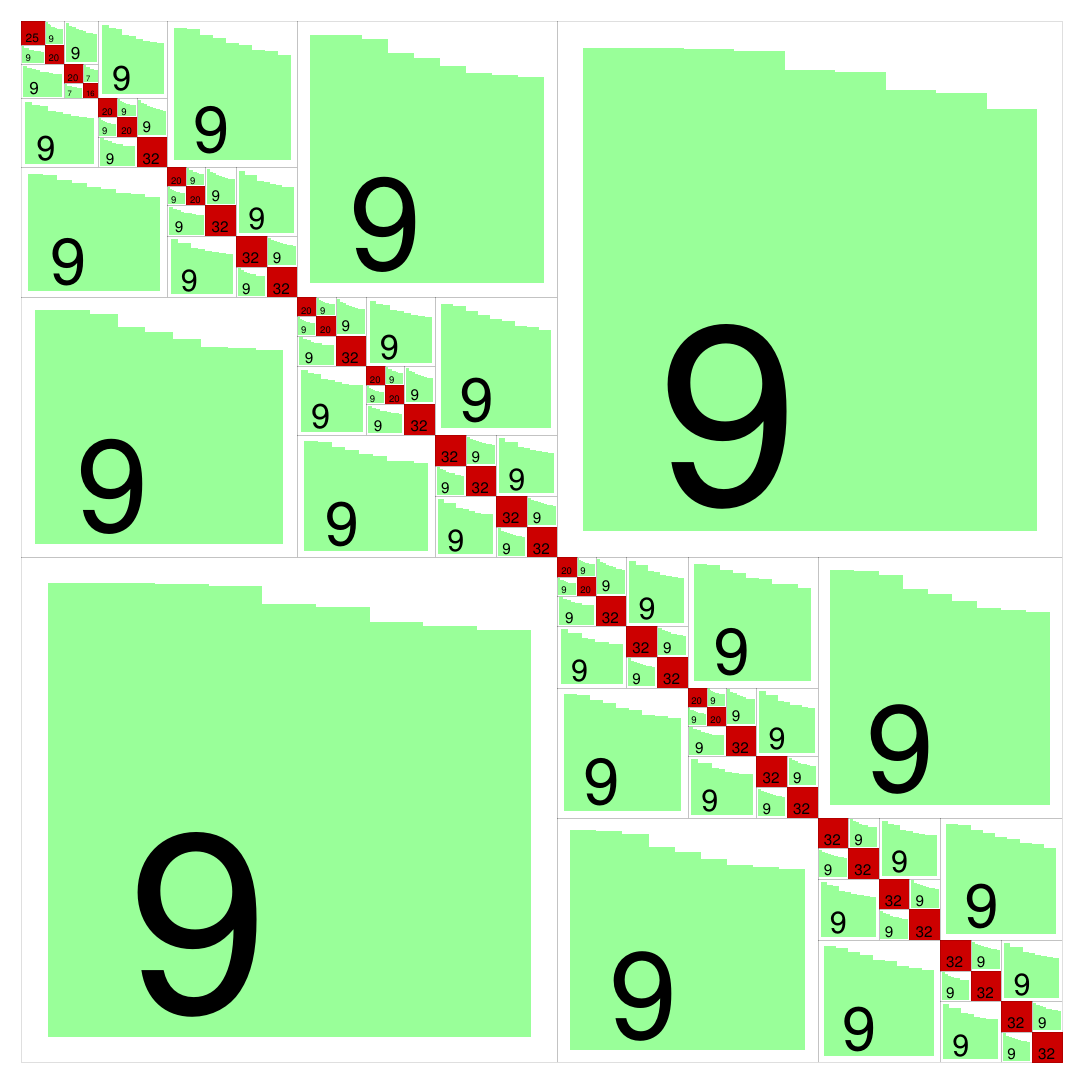}
\caption{Examples of $\H$-matrices with three different admissibility criteria: (left) standard, (middle) domain-decomposition based and (right) weak (HODLR).}
\label{fig:Hexample_adm}
\end{figure}
Figure~\ref{fig:ct_bct} shows examples of cluster trees (1-2) and block cluster trees (3-4). 
\begin{figure}[htbp!]
\centering
\includegraphics[width=0.3\textwidth]{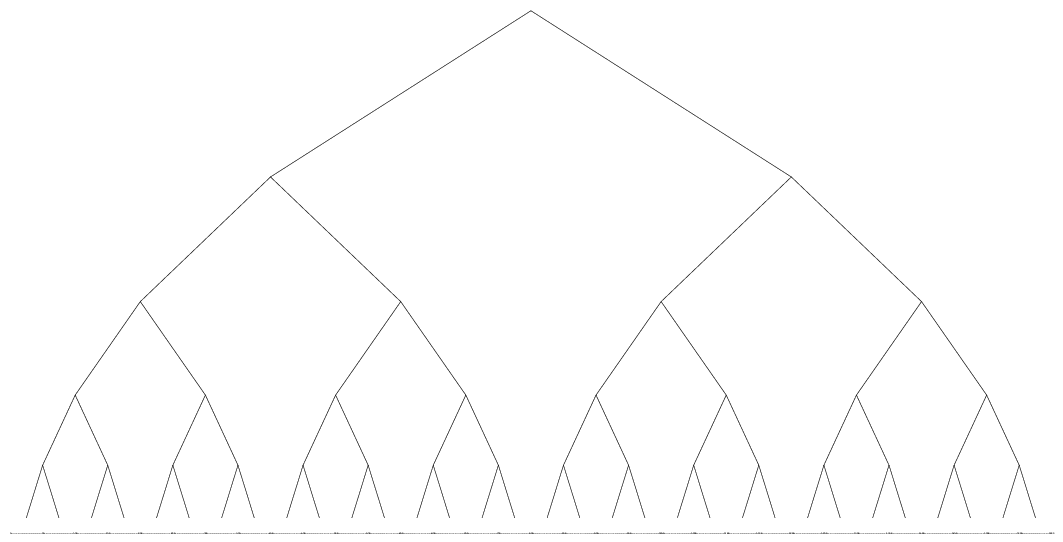}
\includegraphics[width=0.3\textwidth]{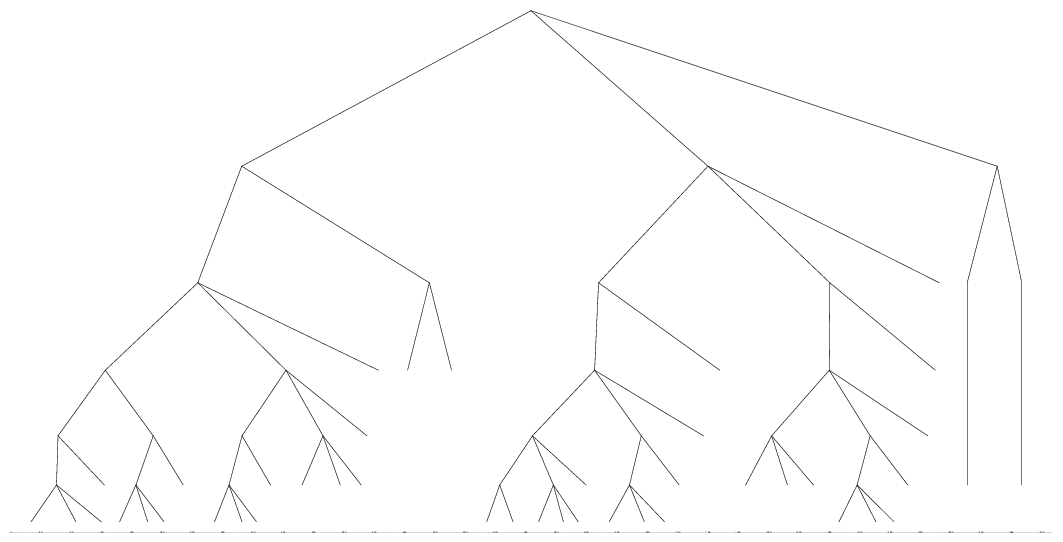}
\includegraphics[width=0.15\textwidth]{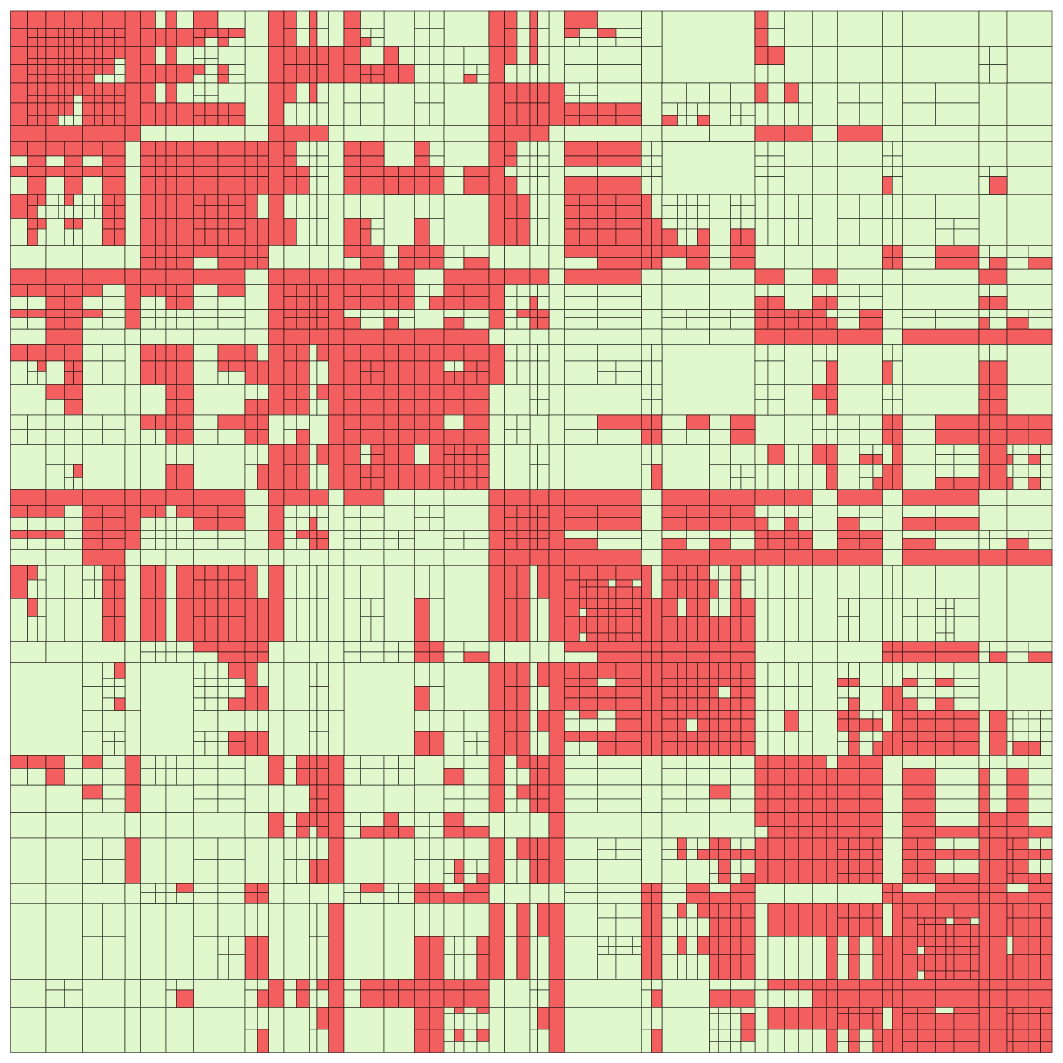}
\includegraphics[width=0.15\textwidth]{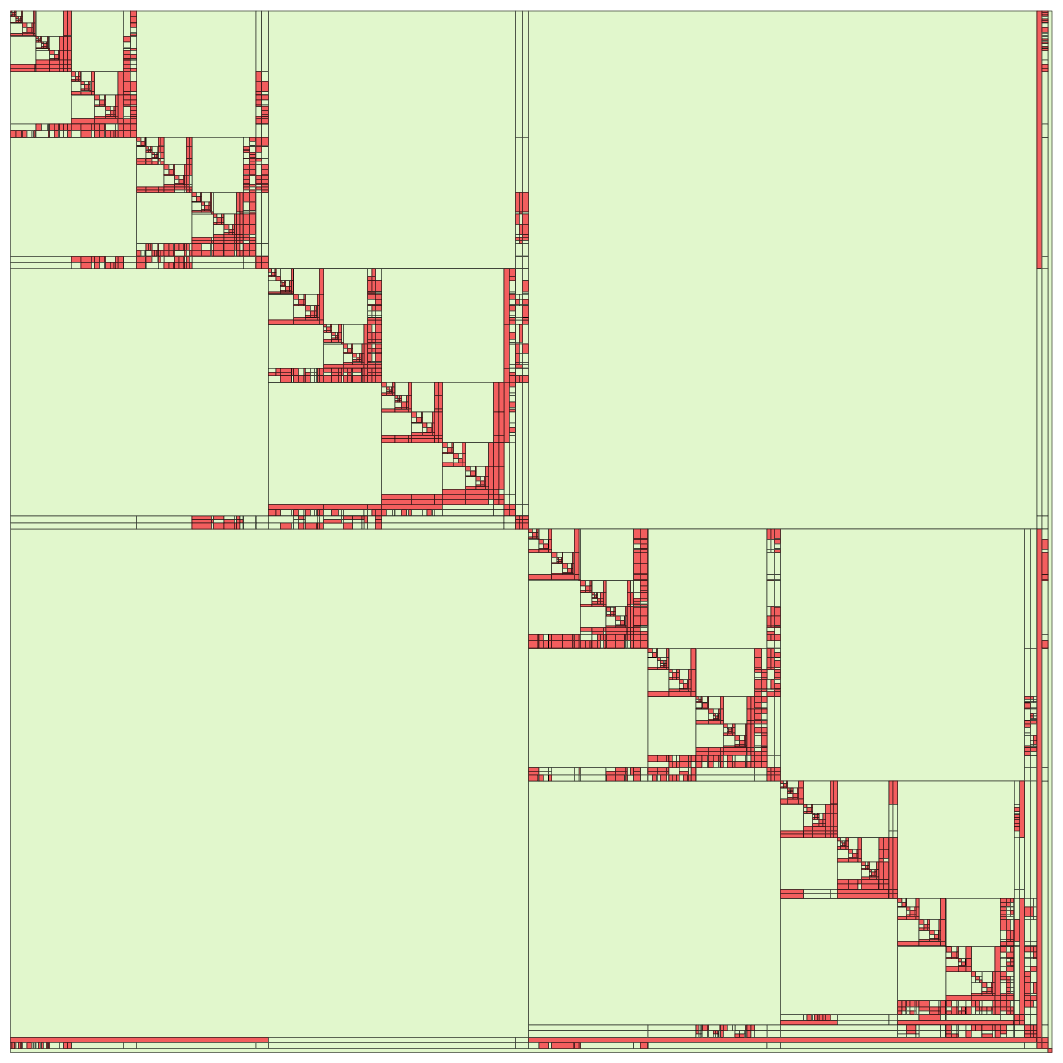}
\caption{(1-2) Two examples of cluster trees: (1st) standard, (2nd) domain-decomposition based; (3) an example of block cluster tree, (4) a block cluster tree, constructed from domain decomposition based cluster tree shown in (2nd).}
\label{fig:ct_bct}
\end{figure}
Let $I$ be an index set of all locations. Denote for each index $i \in I$ corresponding to a basis function $b_i$ (e.g., the ``hat'' function) the support $\mathcal{G}_i := \supp b_i \subset \RR^d$, where $d\in \{1,2,3\}$ is spatial dimension. Now we define two trees which are necessary for the definition of hierarchical matrices. These  trees are labeled trees where the label of a vertex $t$ is denoted by $\hat{t}$.
\begin{defi}(Cluster Tree $T_{I}$)\cite{Part1, GH03}\\
A finite tree $T_I$ is a cluster tree over the index set $I$ if the following conditions hold:
\begin{itemize}
\item $I$ is the root of $T_I$ and a subset $\hat{t} \subseteq I$ holds for all $t \in T_I$. 
\item If $t \in T_I$ is not a leaf, then the set of sons, $\mbox{sons}(t)$, contains disjoint subsets of $I$ and the subset $\hat{t}$ is the disjoint union of its sons, $\hat{t}=\displaystyle{\bigcup_{s \in \mbox{sons}(t)}{\hat{s}}}$.\\
\item If $t \in T_I$ is a leaf, then $\vert \hat{t} \vert \leq n_{min}$ for a fixed number $n_{min}$.
\end{itemize} 
\label{def:ClTr}
\end{defi}
We generalise $\mathcal{G}_i$ to clusters $\tau \in T_I$ by setting $\mathcal{G}_{\tau} := \bigcup_{i\in \tau} \mathcal{G}_i$, 
i.e., $\mathcal{G}_{\tau}$ is the minimal subset of $\RR^d$ that contains the supports of all basis functions $b_i$ with $i \in \tau$.
\begin{defi} (Block Cluster Tree $T_{I\times I}$) \cite{Part1, GH03}\\
\label{def:BLClTree}
Let $T_{I}$ be a cluster tree over the index set $I$. A finite tree $T_{I\times I}$ is a block cluster tree based on $T_I$ if the following conditions hold:
\begin{itemize}
\item $\mbox{root}(T_{I\times I})=I \times I$.
\item Each vertex $b$ of $T_{I\times I}$ has the form $b=(\tau,\sigma)$ with clusters $\tau,\sigma \in T_I$.
\item For each vertex $(\tau,\sigma)$ with $\mbox{sons}(\tau,\sigma)\neq \varnothing$, we have
\begin{equation*}  \label{eq:sons}
\mbox{sons}(\tau,\sigma)=
\left\lbrace
\begin{array}{cl}
{(\tau,\sigma^{'}) : \sigma' \in \mbox{sons}(\sigma)}, &\text{if sons}(\tau)= \varnothing \wedge \text{sons}(\sigma) \neq \varnothing \\ 
{(\tau^{'},\sigma) : \tau^{'} \in \text{sons}(\tau)}, &\text{if sons}(\tau)\neq \varnothing \wedge \text{sons}(\sigma)= \varnothing \\ 
{(\tau',\sigma^{'}) : \tau' \in \text{sons}(\tau), \sigma^{'} \in \text{sons}(\sigma)}, &\text{otherwise}  
\end{array}
\right.
\end{equation*}
\item The label of a vertex $(\tau,\sigma)$ is given by $\widehat{(\tau,\sigma)}=\widehat{\tau} \times \widehat{\sigma} \subseteq I \times I$.
\end{itemize}  
\end{defi}
We can see that $\widehat{root(T_{I \times I})}=I \times I$. This implies that the set of leaves of $T_{I \times I}$ is a partition of $I \times I$.
\begin{defi}
The standard admissibility condition (Adm$_{\eta}$) for two domains $B_{\tau}$ and $B_{\sigma}$ (which actually correspond to two clusters $\tau$ and $\sigma$) is defined as follows
\begin{equation*}
\label{eq:stand_cond}
\min\{\diam(B_{\tau}), \diam(B_{\sigma})\} \leq \eta \dist(B_{\tau} , B_{\sigma}),
\end{equation*}
\end{defi}
where $B_{\tau},\, B_{\sigma} \subset \RR^d$ are axis-parallel bounding boxes of the clusters $\tau$ and $\sigma$ such that $\mathcal{G}_{\tau} \subset B_{\tau}$ and $\mathcal{G}_{\sigma} \subset B_{\sigma}$. $\diam$ and $\dist$ are usual diameter and distance, by default $\eta=  2.0$.
%
%
\section{Appendix: Accuracy Stability Plots}
\label{app:B}
\begin{figure}[htbp!]
    \begin{subfigure}[b]{0.33\textwidth}
     \centering
        \caption{}
       \includegraphics[width=0.99\textwidth]{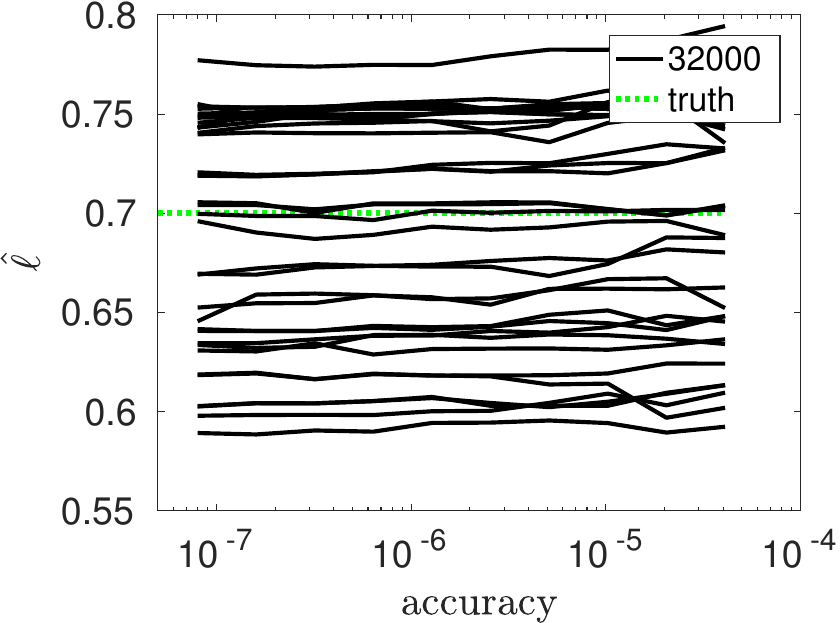}\vspace{-.1cm}
        \label{fig:ell64}
    \end{subfigure}
   \begin{subfigure}[b]{0.32\textwidth}
     \centering
        \caption{}
           \includegraphics[width=0.99\textwidth]{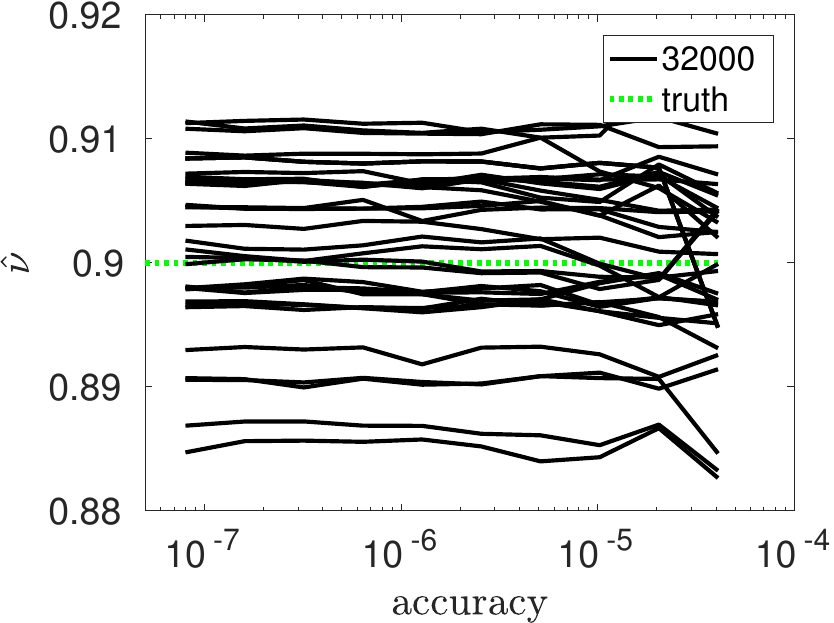}\vspace{-.1cm}
        \label{fig:nu64}
    \end{subfigure}
   \begin{subfigure}[b]{0.32\textwidth}
     \centering
         \caption{}
            \includegraphics[width=0.99\textwidth]{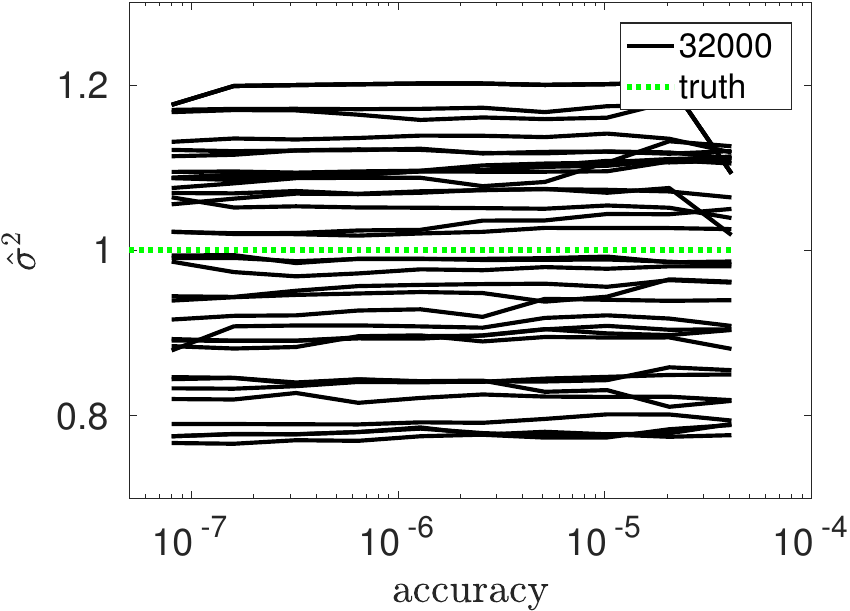}\vspace{-.1cm}
        \label{fig:sigma64}
    \end{subfigure}
\caption{Estimated parameters as a function of the accuracy $\varepsilon$, based on 30 replicates (black solid curves) with $n=64{,}000$ observations. True parameters $\btheta=(\ell,\nu,\sigma^2)=(0.7,0.9.1.0)$ represented by the green doted lines.
Replicates on (a) identify $\hat{\ell}$; on (b) identify $\hat{\nu}$; and on (c) identify $\hat{\sigma}^2$.}
\label{fig:50eps}
\end{figure}
%

%
\end{appendices}
\end{document}